\documentclass{lmcs}
\usepackage[hidelinks]{hyperref}
\usepackage{amsmath}
\usepackage{cleveref}
\usepackage{bbm}
\usepackage{multicol}
\usepackage{xcolor}
\usepackage{tikz}
\usepackage{tikz-cd}
\usepackage{amssymb}
\usepackage{amsthm}
\usepackage{adjustbox}

\usetikzlibrary{calc}

\tikzset{curve/.style={settings={#1},to path={(\tikztostart)
			.. controls ($(\tikztostart)!\pv{pos}!(\tikztotarget)!\pv{height}!270:(\tikztotarget)$)
			and ($(\tikztostart)!1-\pv{pos}!(\tikztotarget)!\pv{height}!270:(\tikztotarget)$)
			.. (\tikztotarget)\tikztonodes}},
	settings/.code={\tikzset{quiver/.cd,#1}
		\def\pv##1{\pgfkeysvalueof{/tikz/quiver/##1}}},
	quiver/.cd,pos/.initial=0.35,height/.initial=0}

\tikzset{tail reversed/.code={\pgfsetarrowsstart{tikzcd to}}}
\tikzset{2tail/.code={\pgfsetarrowsstart{Implies[reversed]}}}
\tikzset{2tail reversed/.code={\pgfsetarrowsstart{Implies}}}
\tikzset{no body/.style={/tikz/dash pattern=on 0 off 1mm}}

\definecolor{RED}{HTML}{9B1B30}
\newcommand{\linkLEM}[1]{\textcolor{RED}{{\hypersetup{linkcolor=RED}\hyperref[defs_lem]{#1}}}}
\newcommand{\linkAC}[1]{\textcolor{teal}{{\hypersetup{linkcolor=teal}\hyperref[defs_ac]{#1}}}}
\newcommand{\linkBLEM}[1]{\textcolor{RED}{{\hypersetup{linkcolor=RED}\hyperref[defs_blem]{#1}}}}
\newcommand{\linkBAC}[1]{\textcolor{teal}{{\hypersetup{linkcolor=teal}\hyperref[defs_bac]{#1}}}}

\newcommand{\LS}{\mathsf{DLS}}
\newcommand{\DLS}{\mathsf{DLS}}
\newcommand{\OAC}{\linkAC{\mathsf{OAC}}}
\newcommand{\AC}{\linkAC{\mathsf{AC}}}
\newcommand{\DC}{\linkAC{\mathsf{DC}}}
\newcommand{\CC}{\linkAC{\mathsf{CC}}}
\newcommand{\OBDC}{\linkBAC{\mathsf{OBDC}}}
\newcommand{\BDC}{\linkBAC{\mathsf{BDC}}}
\newcommand{\DDC}{\linkBAC{\mathsf{DDC}}}
\newcommand{\BCC}{\linkBAC{\mathsf{BCC}}}
\newcommand{\BAC}{\linkBAC{\mathsf{BAC}}}
\newcommand{\DP}{\linkLEM{\mathsf{DP}}}
\newcommand{\DDP}{\linkLEM{\mathsf{EP}}}
\newcommand{\BDP}{\linkBLEM{\mathsf{BDP}}}
\newcommand{\BDDP}{\linkBLEM{\mathsf{BEP}}}
\newcommand{\LEM}{\linkLEM{\mathsf{LEM}}}
\newcommand{\LPO}{\linkLEM{\mathsf{LPO}}}
\newcommand{\MP}{\linkLEM{\mathsf{MP}}}
\newcommand{\KS}{\linkLEM{\mathsf{KS}}}
\newcommand{\GKS}{\linkBLEM{\mathsf{GKS}}}
\newcommand{\GMP}{\linkBLEM{\mathsf{GMP}}}
\newcommand{\IP}{\linkLEM{\mathsf{IP}}}
\newcommand{\BIP}{\linkBLEM{\mathsf{BIP}}}

\newcommand{\M}{\mathcal{M}}
\newcommand{\N}{\mathcal{N}}

\newcommand{\Prop}{\mathbb{P}}
\newcommand{\Nat}{\mathbb{N}}
\newcommand{\Bool}{\mathbb{B}}
\newcommand{\Unit}{\mathbbm{1}}
\newcommand{\Term}{\mathbb{T}}
\renewcommand{\Form}{\mathbb{F}}

\newcommand{\btrue}{\mathsf{true}}
\newcommand{\bfalse}{\mathsf{false}}

\renewcommand{\phi}{\varphi}
\newcommand{\sto}{{\rightarrow}}

\newcommand{\total}[1]{\mathsf{tot}(#1)}

\newcommand{\directed}[1]{\mathsf{dir}(#1)}

\newcommand{\Th}[1]{\mathsf{Th}(#1)}
\newcommand{\x}{\mathsf{x}}

\newcommand{\arity}[1]{\mathalpha{|{#1}|}}
\usepackage{graphicx}
\usepackage{overpic}
\usepackage{cleveref}

\newtheorem{conjecture}{Conjecture}
\usepackage[countsame]{coqtheorem_LMCS}

\hypersetup{
	breaklinks=true,
}

\newcoqtheorem{theorem}{Theorem}
\newcoqtheorem{lemma}{Lemma}
\newcoqtheorem{coqfact}{Fact}
\newcoqtheorem{corollary}{Corollary}
\newcoqtheorem{definition}{Definition}
\newcoqtheorem{coqconjecture}{Conjecture}
\setCoqFilename{toc}

\crefformat{dummy}{#2}

\Crefname{theorem}{Theorem}{Theorems}\Crefname{theoremAux}{Theorem}{Theorems}
\Crefname{lemma}{Lemma}{Lemmas}\Crefname{lemmaAux}{Lemma}{Lemmas}
\Crefname{coqfact}{Fact}{Facts}\Crefname{coqfactAux}{Fact}{Facts}
\Crefname{corollary}{Corollary}{Corollaries}\Crefname{corollaryAux}{Corollary}{Corollaries}
\Crefname{definition}{Definition}{Definitions}\Crefname{definitionAux}{Definition}{Definitions}
\Crefname{dummy}{Definition}{Definitions}
\crefname{dummy}{Definition}{Definitions}

\crefname{theorem}{Theorem}{Theorems}\crefname{theoremAux}{Theorem}{Theorems}
\crefname{lemma}{Lemma}{Lemmas}\crefname{lemmaAux}{Lemma}{Lemmas}
\crefname{coqfact}{Fact}{Facts}\crefname{coqfactAux}{Fact}{Facts}
\crefname{corollary}{Corollary}{Corollaries}\crefname{corollaryAux}{Corollary}{Corollaries}
\crefname{definition}{Definition}{Definitions}\crefname{definitionAux}{Definition}{Definitions}

\newcommand{\origin}{https://dominik-kirst.github.io/bdp/}
\setBaseUrl{{\origin Undecidability.FOL.ModelTheory.}}

\begin{document}

\title[Blurred Drinker Paradoxes and Blurred Choice Axioms]{Blurred Drinker Paradoxes and Blurred Choice Axioms: Constructive Reverse Mathematics of the\texorpdfstring{\\}{}Downward Löwenheim-Skolem Theorem}
\thanks{The first author received funding from the European Union’s Horizon research and innovation programme under the Marie Skłodowska-Curie grant agreement No.101152583 and a Minerva Fellowship of the Minerva Stiftung Gesellschaft für die Forschung mbH}

\author{Dominik Kirst\lmcsorcid{0000-0003-4126-6975}}[a,b]
\author{Haoyi Zeng\lmcsorcid{0009-0007-2506-3787}}[c]
\address{Inria Paris, France}
\email{dominik.kirst@inria.fr}
\address{Ben-Gurion University, Israel}
\address{Saarland University, Germany}
\email{haoyi.per@gmail.com}

\begin{abstract}
	In the setting of constructive reverse mathematics, we analyse the downward Löwenheim-Skolem (DLS) theorem of first-order logic, stating that every infinite model has a countable elementary submodel.
	Refining the well-known equivalence of the DLS theorem to the axiom of dependent choice (DC) over classical base theories, our constructive approach allows for several finer logical decompositions:
	Just assuming countable choice (CC), the DLS theorem is equivalent to the conjunction of DC with a newly identified fragment of the excluded middle (LEM) that we call the blurred drinker paradox (BDP).
	Further without CC, the DLS theorem is equivalent to the conjunction of BDP with similarly blurred weakenings of DC and CC.
	Independently of their connection with the DLS theorem, we also study BDP and the blurred choice axioms on their own, for instance by showing that BDP is LEM without a contribution of Markov's principle and that blurred DC is DC without a contribution of CC.
	The paper is hyperlinked with an accompanying Rocq development.
\end{abstract}

\keywords{Constructive reverse mathematics, Drinker paradox, Dependent choice, Löwenheim-Skolem theorem, Rocq}

\maketitle

\section{Introduction}

The Löwenheim-Skolem theorem\footnote{Usually attributed to Löwenheim~\cite{lowenheim1915moglichkeiten} and Skolem~\cite{Skolem1920-SKOLUB} by name, but credit is also due to Maltsev~\cite{Mal36} who in turn credits Tarski.} is a central result about first-order logic, entailing that the formalism is incapable of distinguishing different infinite cardinalities.
In particular the theorem's so-called downward part, stating that every infinite model (over a countable signature) can be turned into a countably infinite model with otherwise the exact same behaviour, can be considered surprising or even paradoxical:\footnote{Discovered and discussed by Skolem~\cite{skolem1922einige}. See also the discussion by McCarty and Tennant~\cite{mccarty1987skolem} for a constructivist perspective.} even axiom systems like ZF set theory, concerned with uncountably large sets like the reals or their iterated power sets, admit countable interpretations.
This seeming contradiction in particular and its metamathematical relevance in general led to an investigation of the exact logical assumptions under which the downward Löwenheim-Skolem (DLS) theorem applies.

From the perspective of (classical) reverse mathematics~\cite{friedman1976systems,simpson2009subsystems}, there is a definite answer: the DLS theorem is equivalent to the dependent choice axiom (DC), a weak form of the axiom of choice, stating that there is a path through every total relation~\cite{boolos2002computability,espindolalowenheim,karagila2014downward}.
To argue the first direction, one can organise the usually iterative construction of the countable submodel such that a single application of DC yields the desired result.
For the converse direction, one uses the DLS theorem to turn a given total relation $R$ into a countable sub-relation $R'$, applies a (classically available) instance of the axiom of countable choice (CC) to obtain a path $f'$ through $R'$, which is then reflected back as a path $f$ through $R$.
In total, that is:
\vspace{0.2cm}
$$\DLS~\leftrightarrow~\mathsf{DC}$$
\vspace{-0.2cm}

However, the classical answer is insufficient if one investigates the computational content of the DLS theorem, i.e.~the question how effective the transformation of a model into a countable submodel really is.
A more adequate answer can be obtained by switching to the perspective of \emph{constructive} reverse mathematics~\cite{ishihara_reverse_2006,diener2018constructive}, which is concerned with the analysis of logical strength over a constructive meta-theory, i.e.~in particular without the law of excluded middle (LEM), stating that $p\lor\neg p$ for all propositions $p$, and ideally also without CC~\cite{richman2001constructive}.
In that setting, finer logical distinctions become visible and one can analyse the computational content of the DLS theorem by investigating two questions:
\begin{enumerate}
	\item Does the DLS theorem still follow from DC alone, without any contribution of LEM?
	
	\item Does the DLS theorem still imply all of DC, without any contribution of CC?
\end{enumerate}
\vspace{0.2cm}
$$\DLS\,(+\,\mathsf{CC})~\overset{?}\leftrightarrow~\mathsf{DC}\,(+\,\mathsf{LEM})$$
\vspace{-0.2cm}

In this paper, after giving a fully constructive proof of a weak form of the DLS theorem sharing the same computational content as constructivised model existence theorems~\cite{herbelin2024analysis,forster2020completeness}, we observe that neither 1) nor 2) is the case.
Instead, we clarify which exact fragment of LEM is needed on top of DC to prove the DLS theorem and, conversely, which exact fragment of DC it implies.

Regarding (1), note that the DLS theorem requires LEM in the form of the drinker paradox:\footnote{Popularised as a logic puzzle by Smullyan~\cite{smullyan1990name} and studied in relation to other principles of constructive mathematics by Escard\'o and Oliva~\cite{escardo2010searchable}.} in every (non-empty) bar there is a particular person, such that if this person drinks, then everybody in the bar drinks.
The classical explanation for that phenomenon is simple, either everyone drinks anyway, in which case we can choose just an arbitrary person, or there is someone not drinking, in which case we choose that person and obtain a contradiction to the assumption they would drink.
The role of the drinker paradox in the proof of the DLS theorem now is to ensure that the constructed model correctly interprets universal quantification:\footnote{Incidentally, a similar requirement is needed in Henkin-style completeness proofs~\cite{henkin_completeness_1949}, connecting to our favoured strategy to establish the DLS theorem.} given a formula $\forall x.\,\phi(x)$ one can find a special domain element $a$ such that $\phi(a)$ implies $\forall x.\,\phi(x)$, thereby reducing a test over the whole domain to a test of a single point and easing the correctness proof.
However, we observe that one actually does not need to know $a$ concretely but only that it is contained somewhere in the countable model we construct, more formally, that there is a countable subset $A$ such that $\forall a\in A.\,\phi(a)$ implies $\forall x.\,\phi(x)$.
Seen computationally, this means that we reduce testing over the whole domain to testing only a countable part of it.

On a more abstract level, this observation corresponds to a constructively weaker form of the drinker paradox: in every bar, there is a countable group, such that if everyone in this group drinks, then everybody in the bar drinks.
We call this principle the \emph{blurred drinker paradox} as it continues the bar situation at a later point when everyone's vision got blurred and clear identifications of persons become impossible.
That it corresponds to the DLS theorem is suggestive since both statements in a sense collapse arbitrary to countable cardinality and indeed we can show that, with CC still assumed in the background, the DLS theorem is equivalent to the conjunction of DC with the blurred drinker paradox.
On top of this equivalence, we study the principle (and its dual needed for existential quantification) in a more general setting with arbitrary blurring cardinalities and in relation to other non-constructive principles, unveiling a hierarchy of classically invisible logical structure.

Turning to question (2), we observe that DC becomes underivable from the DLS theorem if we further give up on CC in the background.
This suggests that the actual fragment of DC at play is a weakening without the contribution of CC, i.e.~a principle that follows from DC but does not imply CC.
By a deeper analysis of the proof of the DLS theorem, we actually identify several weakenings of DC that happen to include similar blurring techniques as in the case of the blurred drinker paradox, again connected to the indistinguishability of countable and uncountable cardinalities.
In particular, we show that the DLS theorem is equivalent to the conjunction of a strong blurred form of DC and the blurred drinker paradox, with the former further decomposing into a weaker blurred form of DC conjoined with a blurred form of CC.

Orthogonal to its use for the constructive reverse analysis of the DLS theorem, our discussion of blurred choice axioms contributes to the constructive understanding of the logical structure of choice principles in general, thereby complementing related work by Brede and Herbelin~\cite{brede2021logical}.
For instance, we show that in the absence of CC, the core of DC actually states that every total relation has a total countable sub-relation or, alternatively, that every directed relation has a directed countable sub-relation.
These and similar classically equivalent but constructively weaker reformulations of DC are in clear connection to the DLS theorem, compressing some given relational structure to countable cardinality.

Our resulting decomposition may then be depicted in the following way

\begin{figure}[ht]
	\centering
	\begin{overpic}[scale=0.17]{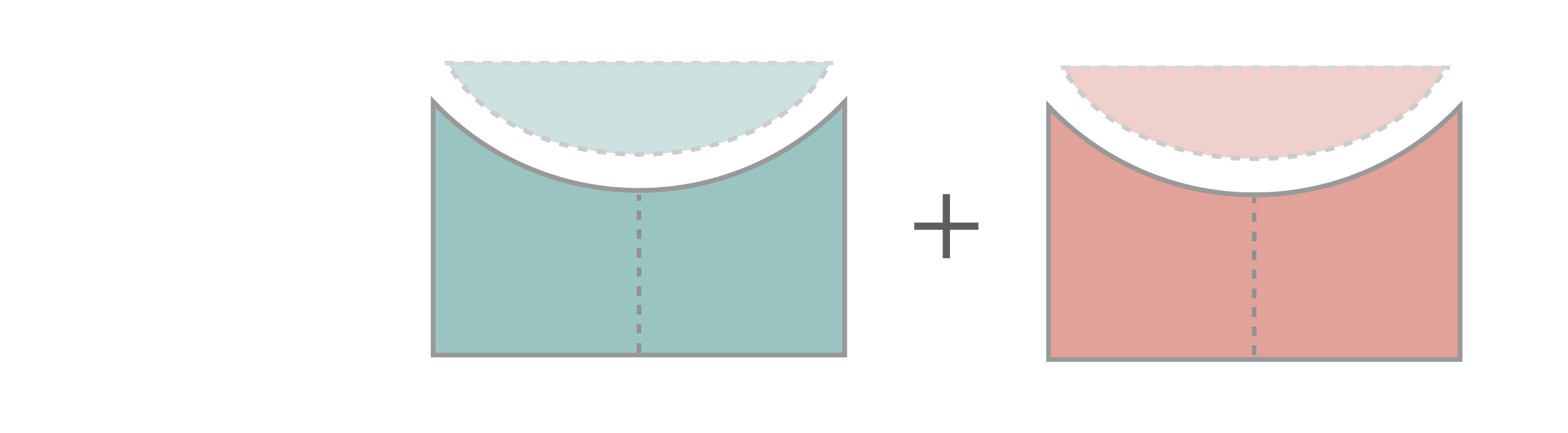}
	  \put(0, 13){\Large$\DLS~~\leftrightarrow$}
	  \put(30, 9){$\BCC$}
	  \put(42.5, 9){$\DDC$}
	  \put(36.8, 19.5){$\CC_{\Nat}$}
	  \put(69, 9){$\BDP$}
	  \put(82, 9){$\BDDP$}
	  \put(77, 19){$\MP$}
	  \put(36.5, -1){\Large$\DC$}
	  \put(74, -1){\Large$\LEM$}
	\end{overpic}
  \end{figure}
  
\noindent
which illustrates that DLS is equivalent to two independent components of DC (abbreviated BCC and DDC) orthogonal to CC on $\Nat$, in addition to two independent components of LEM (abbreviated BDP and BEP), orthogonal to Markov's principle (MP).
Note that the colour-coded abbreviations of all logical principles here and in the remainder of the text are hyperlinked with their definitions in the appendix, where also a more complete diagram of logical connections is given.

While the present paper is written in a deliberately informal way to comply with many systems of (higher-order) constructive mathematics and to address a broad audience, we complement it with a fully mechanised development using the Rocq proof assistant~\cite{the_rocq_development_team_2025_17473943}, formerly named Coq.
That is, all definitions and theorems have been formalised in the type theory implemented in Rocq such that the correctness of all proofs can be machine-checked.
The reasons we do this and actually find it worthwhile are threefold:
First, the mechanisation guarantees that all constructions and arguments are correct, which is especially helpful for intricate syntactical arguments needed in the proof of the DLS theorem.
Secondly, using a proof assistant actually helped us identify the new non-constructive principles at play by pointing to the constructions and proofs that needed modification.
Thirdly, as employing the Curry-Howard isomorphism proving in Rocq is the same as programming, the computational content of constructive proofs is made explicit: for instance, the fully constructive proof of the weak DLS theorem, in principle, computes the countable submodel, cf.~\cite{herbelin2024analysis}.

\paragraph*{Contributions}

The contributions of this paper are:

\begin{itemize}
	\item
	We introduce the blurred drinker paradox and blurred choice axioms as natural families of logical principles in the context of constructive reverse mathematics.
	To classify their strength, among others we show that the blurred drinker paradox is LEM without a contribution of Markov's principle (\Cref{fact_BDP_LEM}) and that the blurred forms of DC are DC without a contribution of CC (\Cref{larry_bdc2}).
	\item
	Using these logical principles, we give precise constructive decompositions of the DLS theorem: assuming CC, it is equivalent to DC and the blurred drinker paradox (\Cref{larry_blurred}), and without CC, the same equivalence holds for various blurrings of DC and CC (\Cref{thm_full}).
	Moreover, we observe that a weak form of the DLS theorem is fully constructive (\Cref{fact_LS_compl}), as a by-product of a known fully constructive model existence theorem~\cite{herbelin2024analysis,forster2020completeness}.
	\item
	Our underlying proof strategy for the DLS theorem (\Cref{thm_LS_environment}) is a streamlining of usual textbook proofs: we construct a syntactic model and collect all structural information in variable environments.
	Thereby the proof relies neither on signature nor domain extensions and is particularly suitable for computer mechanisation.
	\item 
	Our paper is accompanied by a Rocq development,\footnote{As part of the Rocq-FOL library: \url{https://github.com/uds-psl/coq-library-fol}} ensuring the correctness of our proofs and providing full formal detail, such that the text may remain on a more accessible level.
	For seamless integration, all formalised definitions and theorems in the PDF version of this paper are hyperlinked with \href{\origin toc.html}{HTML documentation} of the code.
	Only the sketched separations in \Cref{sec_separate} remain unformalised.
	\item
	We correct an apparent oversight in the investigation of sub-classical logical principles:\footnote{For instance, a relevant file in the Rocq standard library (\url{https://coq.inria.fr/doc/v8.20/stdlib/Coq.Logic.ClassicalFacts.html}) refers to both the drinker paradoxes and the independence of premise as principles strictly weaker than LEM, which is actually only the case if one fixes a domain in advance.} in a higher-order logic, the universal closures of the drinker paradox, the existence principle, and the independence of (general) premise are all equivalent to LEM (\Cref{fact_dp_lem}).
	\item
	In comparison to the conference version~\cite{kirst2025blurred} of this paper, aside from improved explanations and comments at various places, we provide a more general discussion of the blurred axioms, also in comparison with the DLS theorem for arbitrary signature cardinalities (\Cref{sec_general}), and report on some initial results separating them from each other and from their full counterparts (\Cref{sec_separate}).
\end{itemize}

\paragraph*{Outline}
\Cref{sec_prelims} provides an overview of some standard non-constructive axioms and basic concepts of first-order logic.
In \Cref{sec_LS_constructive}, we present three constructive versions of the $\LS$ theorem of increasing strength and, in \Cref{sec_LS_classical}, we reconstruct the classical equivalence of the $\LS$ theorem to DC.
This equivalence is then refined by introducing the blurred drinker paradox in \Cref{sec_BDP}, used in \Cref{sec_LS_blurred} to replace the use of LEM, and by introducing blurred choice axioms in \Cref{sec_BAC}, used in \Cref{sec_full} to replace the use of DC.
In the present extended version of the paper, these results on the $\LS$ theorems are complemented by generalisations of the blurred axioms in \Cref{sec_general} and some independence proofs in \Cref{sec_separate}.
We close with a discussion concerning the main results, the Rocq mechanisation, and future work in \Cref{sec_discussion}.
Note that \Cref{sec_BDP,sec_BAC,sec_general,sec_separate} are written to be accessible for readers only interested in the new logical principles and their decompositions, independent of their use for the $\LS$ theorem in the other sections.

\section{Preliminaries}
\label{sec_prelims}

\subsection{Meta-Theory}

We work in a constructive meta-theory that we leave underspecified to generalise over concrete standard systems such as intuitionistic higher-order arithmetics like HA$^\omega$, intuitionistic or constructive set theories like IZF and CZF, and constructive type theories like MLTT, HoTT, and CIC.
Of course, the latter referring to the Calculus of inductive Constructions~\cite{CoquandCC,paulin1993inductive} implemented in the Rocq proof assistant~\cite{the_rocq_development_team_2025_17473943} is the concretisation we have in mind, so we also lean towards some type-theoretic notation and jargon, in which we mechanised all results presented in this paper.

On the logical level, we stipulate an impredicative collection $\Prop$ of propositions with standard notation $(\bot,\top,\neg,\land,\lor,\forall,\exists)$ to express composite formulas and a means to include inductively defined predicates.
On the computational level, we assume collections like $\Nat$ of natural numbers and $\Bool$ of Booleans, function spaces like $\Nat\sto\Bool$, and a means to include inductively defined collections.

We frequently use a \emph{pairing function} encoding pairs $(n, m):\Nat^2$ as numbers $\langle n,m \rangle:\Nat$.
We write $f\,\langle n,m \rangle:=\dots$ for function definitions treating an input as an encoded pair.

Given $A$, if there are functions $i:A\sto \Nat$ and $j:\Nat\sto A$ with $j\,(i\,x)=x$ for all $x:X$, then we say that $A$ is \emph{countable}, where we in particular include finite $A$ to avoid speaking of \emph{at most} countable models in the formulations of the $\LS$ theorem.
Note that there are many non-equivalent definitions of countability in constructive logic but for our purposes most of them would do.
Similarly, we represent \emph{countable subsets} as functions $f,g:\Nat \sto A$, and, slightly abusing set-theoretic notation, write $f\subseteq g$ if for every $n$ there is $m$ with $f\,n=g\,m$ and $f\cup g:\Nat\sto A$ for the subset
\begin{align*}
	(f\cup g)\, (2n)~&:=~f n\\
	(f\cup g)\, (2n+1)~&:=~g n
\end{align*}
satisfying expectable properties like $f\subseteq f\cup g$ and $g\subseteq f\cup g$.

Lastly, we call a predicate $P:A\sto \Prop$ \emph{decidable} if it coincides with a Boolean function $f:A\sto \Bool$, i.e.~if $\forall x:A.\,P\,x\leftrightarrow f\,x=\btrue$.
This definition naturally generalises to binary relations $R:A\sto B\sto \Prop$ and relations of higher arity.

\subsection{Constructive Reverse Mathematics}

The idea of constructive reverse mathematics is to identify non-constructive logical principles and their equivalences to well-known theorems, thereby classifying logical strength and computational content~\cite{ishihara_reverse_2006,diener2018constructive,bridges2023handbook}.
In preparation of upcoming similar results, we review some well-known connections of logical principles like
\begin{align*}
	\LEM ~:=~& \forall p:\Prop.\,p\lor \neg p\\
	\LPO ~:=~& \forall f:\Nat\sto\Bool.\,(\exists n.\, f\,n=\btrue)\lor (\forall n.\,f\,n=\bfalse)\\
	\DP_A ~:=~& \forall P:A\sto\Prop.\,\exists x.\, P\,x \to \forall y.\, P\,y\\
	\DDP_A ~:=~& \forall P:A\sto\Prop.\,\exists x.\, (\exists y.\,P\,y) \to P\,x\\
	\IP_A ~:=~& \forall P:A\sto\Prop.\,\forall p:\Prop.\,(p\to \exists x.\,P\,x)\to \exists x.\,p\to P\,x
\end{align*}
namely the law of excluded middle, the limited principle of omniscience, the drinker paradox, the existence principle, and the independence of (general) premise principle.\footnote{Note that our convention is that quantifiers bind their relevant variable with maximal scope, e.g.\ the body of $\DP_A$ expresses $\exists x.\, (P\,x \to \forall y.\, P\,y)$ and not $(\exists x.\, P\,x) \to \forall y.\, P\,y$.}
In the situation of $\DP_A$ for $P$, we call the given $x$ the \emph{Henkin witness} for $P$, same for $\DDP_A$ which is a dual variant of the drinker paradox.
We write $\DP$ to denote $\DP_A$ for all inhabited $A$, analogously for $\DDP$ and $\IP$, but 
state results in the more localised form where possible.

\setCoqFilename{LogicalPrinciples}

\begin{coqfact}[ ][scheme_facts_2]
	\label{fact_lpo}
	The following statements hold:
	\begin{enumerate}
		\coqitem[DP_nat_impl_LPO]
		Both $\DP_\Nat$ and $\DDP_\Nat$ imply $\LPO$.
		\coqitem[IP_iff_EP]
		$\DDP_A$ is equivalent to $\IP_A$.
	\end{enumerate}
\end{coqfact}

\begin{proof}
	For (1), assuming $\DP_\Nat$ and a function $f:\Nat\sto\Bool$ yields some $n$ such that $f\,n=\bfalse$ implies $f\,n'=\bfalse$ for all $n'$.
	Then the claim follows by case analysis of $f\,n$.
	The claim for $\DDP_\Nat$ follows analogously and (2) is straightforward, with the choice $p:=\exists y.\,P\,y$ for the backward direction.
\end{proof}

In contrast to the situation in first-order logic~\cite{warren2018drinker}, the universal closures of these principles in a higher-order meta-theory with comprehension have the full strength of $\LEM$:

\begin{coqfact}[ ][there_are_equivalent]
	\label{fact_dp_lem}
	$\LEM$, $\DP$, $\DDP$, and $\IP$ are equivalent.
\end{coqfact}

\begin{proof}
	That $\LEM$ implies the other principles is well-known.
	As an example for the converse, assume $\DP$ and some $p:\Prop$.
	Using $\DP$ for $A:=\{b:\Bool\mid b=\bfalse \lor (p\lor \neg p)\}$ and
	$$P\,b~:=~\begin{cases}\neg p&\text{if }b=\btrue\\ \top&\text{otherwise}\end{cases}$$
	yields a Henkin witness $b:A$ for $P$.
	If $b=\btrue$, we directly obtain $p\lor\neg p$ and if $b=\bfalse$, then we derive $\neg p$ as follows:
	On assumption of $p$ we know that $\btrue$ is a member of $A$ and since $P\,b=\top$, by the Henkin property we obtain $P\,b'$ for all $b':A$.
	So for $b':=\btrue$ in $A$ we then obtain $\neg p$, in contradiction to the of assumption $p$.
\end{proof}

While the previous principles concern structure below $\LEM$, we now consider structure below the axiom of choice~\cite{jech2008axiom}
\begin{align*}
	\AC_{A,B} ~:=~& \forall R: A\sto B\sto\Prop.\, \total R \to \exists f:A\sto B. \forall x.\, R\,x\,(f\,x)\\
	\DC_{A} ~:=~&\forall R:A\sto A \sto\Prop.\,\total R\to \exists f.\,\forall n.\,R\,(f\,n)\,(f\,(n+1))\\
	\CC_{A} ~:=~& \forall R: \Nat\sto A\sto\Prop.\, \total R \to \exists f:\Nat\sto A. \forall n.\, R\,n\,(f\,n)\\
	\OAC_{A,B} ~:=~& \forall R: A\sto B\sto\Prop.  \exists f:A\sto B. \,\total R \to \forall x.\, R\,x\,(f\,x)
\end{align*}
where $\total R$ expresses the totality of $R$, i.e.\ that for every $x:A$ there is $y:B$ with $R\,x\,y$.
These principles are the axiom of choice, dependent choice, countable choice, and omniscient choice.
Note that the latter is a combination of $\AC$ and $\IP$, similar combinations work for other choice axioms:

\begin{coqfact}[ ][OAC_impl_AC_IP]
	\label{fact_oac}
	For inhabited $A$, $\OAC_{A,B}$ is equivalent to the conjunction of $\AC_{A,B}$ and $\IP_B$.
\end{coqfact}

\begin{proof}
	That $\OAC_{A,B}$ implies $\AC_{A,B}$ is obvious and to derive $\IP_B$ for $P:B\sto\Prop$ one instantiates $\OAC_{A,B}$ to $R\,x\,y:=P\,y$.
	Conversely deriving $\OAC_{A,B}$ for $R:A\sto B\sto \Prop$, note that just using $\AC_{A,B}$ on $R$ would require $\IP_{A\sto B}$ to allow postponing the totality proof.
	Instead, using
	$$R'\,x\,y ~:=~ (\exists y'.\, R\,x\,y')\to R\,x\,y$$
	we just need $\IP_B$ to show $R'$ total to obtain a choice function $f:A\sto B$ from $\AC_{A,B}$ then also witnessing $\OAC_{A,B}$.
\end{proof}

As for the previous principles, we write $\AC$ to denote $\AC_{A,B}$ for all $A,B$ and analogously for the other choice principles, with the restriction to inhabited $A$ in the case of $\DC$.

\begin{coqfact}[ ][AC_impl_DC]
	\label{fact_ac}
	$\AC$ implies $\DC$ and $\DC$ implies $\CC$.
\end{coqfact}

\begin{proof}
	These follow by well-known arguments, see~\cite{jech2008axiom} for instance.
	We sketch the implication from $\DC$ to $\CC$ to prepare a more general version presented in \Cref{fact_bdc}.
	First note that $\DC$ can be equivalently stated for arbitrary $x_0:A$ as
	$$\forall R:A\sto A \sto\Prop.\,\total R\to \exists f.\,f\,0=x_0\land\forall n.\,R\,(f\,n)\,(f\,(n+1))$$
	by restricting $R$ to the sub-relation $R'$ reachable from $x_0$.
	Then a path $f$ through $R$ induces a path $f'$ through $R'$ by first taking the path from $x_0$ to $f\,0$ and by then continuing with $f$.
	
	Now to show $\CC$, assume a total relation $\Nat\sto A\sto\Prop$ on $A$ with some element $a_0$ and consider $A':=\Nat\times A$ and
	$$R'\,(n,x)\,(m,y)~:=~m = n+1\land R\,n\,y$$
	which is total since $R$ is total.
	The modified version of $\DC$ for $R'$ and the choice $x_0:=(0, a_0)$ then yields a path $f':\Nat\to\Nat\times A$ through $R'$ and it is straightforward to verify that $f\,n:=\pi_2\,(f'\,(n+1))$ is a choice function for $R$.
\end{proof}

We write $\DC^\Delta$ and $\CC^\Delta$ for $\DC$ and $\CC$ restricted to decidable relations, respectively.
For \Cref{fact_delta}, we assume that the concrete instance $\CC^\Delta_\Nat$ holds in our meta-theory, as is the case in most formulations of constructive mathematics, while $\DC^\Delta$ is usually unprovable.

Also, while in (extensional) set-theoretic systems $\AC$ implies $\LEM$, this is not the case in most (intensional) type-theoretic systems, and in neither of those does $\DC$ imply $\LEM$.

\subsection{First-Order Logic}

We summarise the concepts for first-order logic (FOL) needed to state the downward Löwenheim-Skolem ($\DLS$) theorem.
The \emph{syntax} of FOL is represented inductively by terms $t:\Term$ and formulas $\phi:\Form$ depending on signatures of function and relation symbols $f$ and $P$:
\begin{align*}
t:\Term&~::=~\x_n\mid f\,\vec{t}\hspace{9.4em}(n:\Nat)\\
\phi,\psi:\Form&~::=~\dot{\bot}\mid P\,\vec{t}\mid\phi\dot{\to}\psi\mid\phi \dot{\land}\psi\mid\phi\dot{\lor}\psi\mid\dot{\forall}\phi\mid\dot{\exists}\phi\hspace{0.6cm}
\end{align*}
The term vectors $\vec t$ are required to have length matching the specified arities $\arity f$ and $\arity P$ of $f$ and $P$.
The negative fragment of FOL referred to in \Cref{fact_modex,fact_LS_compl} comprises formulas only constructed with $\dot\bot$, $\dot\to$, and $\dot\forall$.
Crucially, we assume that the signatures of function and relation symbols are countable, which induces that so are $\Term$ and $\Form$.

\emph{Variable binding} is expressed using de~Bruijn indices~\cite{de_bruijn_lambda_1972}, where a bound variable is encoded as the number of quantifiers shadowing its relevant binder.
Capture-avoiding instantiation with parallel substitutions $\sigma:\Nat\sto \Term$ is defined both for terms as $t[\sigma]$ and formulas as $\phi[\sigma]$.
Notably, $(\dot\forall\phi)[\sigma]$ is defined by $\dot\forall\phi[\uparrow\!\sigma]$ where $\uparrow\!\sigma$ is a suitable shifting substitution.
We denote by $t:\Term^c$ and $\phi:\Form^c$ the closed terms and formulas, respectively, i.e.~those that do not contain free variables.
The latter are also called sentences.

The standard notion of \emph{Tarski semantics} is obtained by interpreting formulas in models $\M$ identified with their underlying domain, providing interpretation functions $\M^{\arity f}\to \M$ for each $f$ and relations $\M^{\arity P}\to \Prop$ for each $P$.
Given an environment $\rho:\Nat\to \M$ to handle free variables, we define term evaluation $\hat\rho\,t$ and formula satisfaction $\M\vDash_\rho\phi$ recursively.
For instance,
$$\M\vDash_\rho\dot\forall\phi~:=~\forall x:\M.\,\M\vDash_\rho\phi[x]$$
with $\phi[x]$ being a notational shorthand for $\phi$ in the updated environment mapping the first variable to the domain element $x$, so $\M\vDash_\rho\phi[x]$ is short for $\M\vDash_{x;\rho}\phi$.

While we will mostly be concerned with semantic considerations, to illustrate the connection of the downward Löwenheim-Skolem theorem to completeness, we also briefly use \emph{deduction systems}.
Deduction systems are represented by inductive predicates $\Gamma\vdash\phi$ relating predicates $\Gamma:\Form\sto\Prop$ on formulas with derivable formulas $\phi$, for instance by rules in the style of natural deduction.
A classical system is obtained by incorporating a rule like double negation elimination, which in a constructive meta-theory is only sound for classical models, i.e.~models satisfying $\M\vDash_\rho \phi$ or $\M\vDash_\rho\dot\neg\phi$ for all $\phi$ and $\rho$.

\setBaseUrl{{\origin Undecidability.FOL.}}

\setCoqFilename{Semantics.Tarski.FullSoundness}
\begin{coqfact}[Soundness][sound_for_classical_model]
	If $\Gamma\vdash \phi$, then $\M\vDash \phi$ for every classical model $\M$ with $\M\vDash\Gamma$.
\end{coqfact}

\begin{proof}
	By induction on the derivation of $\Gamma\vdash \phi$, most cases are straightforward.
	To show the classical derivation rule sound, the classicality of the model is required.
\end{proof}

The converse property of soundness is completeness, stating that semantic validity implies syntactic provability.
In full generality, completeness cannot be proven constructively~\cite{KreiselMP,henkin1954metamathematical,krivtsov2015semantical,espindola2016semantic,simpson2009subsystems}, but the intermediate model existence theorem is constructive for the negative fragment only containing the logical connectives $\dot{\bot}$, $\dot{\to}$, and $\dot{\forall}$~\cite{herbelin2024analysis,forster2020completeness}.

\setCoqFilename{Completeness.TarskiCompleteness}
\begin{coqfact}[Model Existence][model_bot_correct]
	\label{fact_modex}
	In the negative fragment of FOL, for every consistent context $\Gamma$ of sentences one can construct a syntactic model $\M$ over the domain $\Term$ such that $\M\vDash \Gamma$.
\end{coqfact}

\begin{proof}
	We outline the main construction as it will be relevant for similar syntactic models used in \Cref{fact_LS_witness,thm_LS_environment}.
	In a first step, a constructive version of the Lindenbaum Lemma is used to extend $\Gamma$ into a consistent context $\Delta\supseteq\Gamma$ with suitable closure properties for provability and each considered logical connective.
	Next, a model over domain $\Term$ with
	$$f^{\M}\,\vec t~:=~f\,\vec t\hspace{2em}\text{and}
	\hspace{2em}P^{\M}\,\vec t~:=~P\,\vec t\in \Delta$$
	is constructed, for which the so-called Truth Lemma
	$$\M\vDash_\sigma \phi~\leftrightarrow~\phi[\sigma]\in \Delta$$
	is verified by induction on $\phi$ for all $\sigma:\Nat\sto\Term$, acting both as substitution and environment in $\M$.
	Then since $\Gamma\subseteq\Delta$, in particular $\M\vDash \Gamma$ follows.
\end{proof}

We will see in \Cref{fact_LS_compl} that the model existence theorem yields a weak but fully constructive formulation of the $\DLS$ theorem.
This formulation will be based on the notion of elementary equivalence.

\setBaseUrl{{\origin Undecidability.FOL.ModelTheory.}}
\setCoqFilename{Core}

\begin{definition}[Elementary Equivalence][elementary_equivalence]
	Two models $\M$ and $\N$ are \emph{elementarily equivalent} if they satisfy the same sentences, i.e.~if for every $\phi:\Form^c$ we have $\M\vDash \phi$ iff $\N\vDash \phi$.
\end{definition}

Note that elementarily equivalent models only satisfy the same closed formulas but otherwise may behave extremely differently.
A much stronger formulation of the $\DLS$ theorem is based on elementary embeddings, taking all formulas into account and therefore completely aligning the behaviour of the models.

\begin{definition}[Elementary Submodel][elementary_homomorphism]
	Given models $\M$ and $\N$, we call $h:\M\sto \N$ an \emph{elementary embedding} if
	$$\forall \rho\phi.\,\M\vDash_\rho \phi\leftrightarrow\N\vDash_{h\circ\rho} \phi.$$
	If such an elementary embedding $h$ exists, we call $\M$ an \emph{elementary submodel} of $\N$.
\end{definition}


\begin{definition}[ ][DLS]
	The $\DLS$ theorem states that every (inhabited) model of a countable signature has a countable elementary submodel.
\end{definition}

\section{Constructive Downward Löwenheim-Skolem Theorem}
\label{sec_LS_constructive}

We begin with a comparison of different constructive proof strategies for the $\LS$ theorem at various strengths, mostly to identify the underlying concepts in preparation of upcoming results.
First, a weak formulation only yielding an elementarily equivalent model but not necessarily an elementary submodel is obtained as a by-product of a Henkin-style completeness proof via model existence~\cite{henkin_completeness_1949}.
Since the Henkin construction is fully constructive in the negative fragment~\cite{herbelin2024analysis,forster2020completeness}, so is the derived $\LS$ theorem.

\setCoqFilename{FragmentHenkinModel}
\begin{coqfact}[DLS via Model Existence][completeness_LS]
	\label{fact_LS_compl}
	In the negative fragment of FOL, for every classical model one can construct an elementarily equivalent syntactic model.
\end{coqfact}

\begin{proof}
	Given that $\M$ is classical, we can use soundness to show that the set $\Th \M:=\{\phi:\Form^c\mid \M\vDash \phi \}$ of closed formulas satisfied by $\M$ is consistent.
	Then by model existence (\Cref{fact_modex}), there is a model $\N$ with (countable) domain $\Term$ and $\N\vDash \Th\M$.
	This already establishes the first implication showing $\M$ elementarily equivalent to $\N$. For the converse, assuming a closed formula $\phi$ with $\N\vDash \phi$, we obtain $\M\vDash \phi$ by using the classicality of $\M$ and the observation that, if it were $\M\vDash \dot\neg \phi$ instead, also $\N\vDash \dot\neg\phi$ would follow, contradiction.
\end{proof}

The model existence proof can be extended to the full syntax using $\LEM$ alone~\cite{forster2020completeness}, so the derived version of the $\LS$ theorem notably does not rely on any form of choice axioms.
In fact, already the weak law of excluded middle ($\forall p.\,\neg p\lor \neg\neg p$) is sufficient~\cite{herbelin2023new} but we are not aware of a proof showing it necessary for this form of the $\LS$ theorem.

Also note that the Lindenbaum extension used in the proof of \Cref{fact_modex} ensures that quantified formulas have associated Henkin witnesses in the form of unused variables.
In the second variant, this intermediate step is not necessary, since we restrict to models that address all Henkin witnesses by closed terms.

\setCoqFilename{ConstructiveLS}

\begin{definition}[Witness Property][witness_prop]
	Given a model $\M$ with environment $\rho$, we call $w:\M$ a \emph{Henkin witness} for $\dot\forall\phi$ if
	$$\M\vDash_\rho\phi[w]~\to~\M\vDash_\rho\dot\forall\phi$$
	and, symmetrically, a Henkin witness for $\dot\exists\phi$ if
	$$\M\vDash_\rho\dot\exists \phi~\to~\M\vDash_\rho\phi[w].$$
	We say that $\M$ has the \emph{witness property} if Henkin witnesses for all formulas can be expressed by closed terms $t:\Term^c$.
\end{definition}

For models with the witness property, we can then derive the stronger conclusion yielding a countable elementary submodel by means of a simplified syntactic model construction.

\begin{coqfact}[DLS via Witnesses][WitnessProperty_LS]
	\label{fact_LS_witness}
	For every model with the witness property one can construct a syntactic elementary submodel.
\end{coqfact}

\begin{proof}
	Given $\M$ with the witness property and an arbitrary environment $\rho$, we consider the syntactic model $\N$ constructed over the (countable) domain $\Term$ by setting
	$$f^{\N}\,\vec t~:=~f\,\vec t\hspace{2em}\text{and}
	\hspace{2em}P^{\N}\,\vec t~:=~P^\M\,(\hat \rho\,\vec t).$$
	
	We prove that $\hat \rho$ is an elementary embedding of $\N$ into $\M$, i.e.~that $\N\vDash_{\sigma} \phi$ if and only if $\M\vDash_{\hat\rho\circ\sigma} \phi$ for all $\sigma:\Nat\sto\Term$ and $\phi$ by induction on $\phi$.
	The only cases of interest are the quantifiers, we explain universal quantification as example.
	
	Let $t:\Term^c$ denote the Henkin witness for $\dot\forall \phi$
	and assume $\N\vDash_{\sigma} \dot\forall\phi$.
	Then in particular $\N\vDash_{\sigma}\phi[t]$ and by inductive hypothesis $\M\vDash_{\hat\rho\circ\sigma} \phi[t]$, which implies $\M\vDash_{\hat\rho\circ\sigma} \dot\forall\phi$ by the Henkin property of $t$.
	That conversely $\M\vDash_{\hat\rho\circ\sigma} \dot\forall\phi$ implies $\N\vDash_{\sigma} \dot\forall\phi$ is straightforward.
\end{proof}

Many proofs of the $\LS$ theorem proceed by extending the signature with enough fresh constants such that a model satisfying the witness property can be constructed~\cite{conant_notes}.
Alternatively, as a third variant, we replace the condition to represent Henkin witnesses syntactically with environments collecting them semantically.

\begin{definition}[Henkin Environment][Henkin_env]
	Given a model $\M$, we call $\rho:\Nat\sto\M$ a \emph{Henkin environment} if it collects Henkin witnesses for every formula $\phi$ as follows:
	\begin{align*}
		\exists n.\,\M\vDash_\rho\phi[\rho\,n]~\to~\M\vDash_\rho\dot\forall\phi\\
		\exists n.\,\M\vDash_\rho\dot\exists\phi~\to~\M\vDash_\rho\phi[\rho\,n]
	\end{align*}
\end{definition}

Note that if $\M$ has the witness property, then $\M$ admits a Henkin environment by enumerating the evaluations of closed terms, but not vice versa.

The use of Henkin environments then allows to conclude the $\LS$ theorem without extending the signature or domain, which is a particularly suitable strategy for mechanisation.

\begin{theorem}[DLS via Environments][Henkin_LS]
	\label{thm_LS_environment}
	For every model admitting a Henkin environment one can construct a syntactic elementary submodel.
\end{theorem}

\begin{proof}
	Given a model $\M$ with Henkin environment $\rho$, we proceed as in the previous proof, i.e.~we consider the syntactic model $\N$ induced by $\rho$.
	Again, inductively verifying that $\hat \rho$ is an elementary embedding of $\N$ into $\M$ is only non-trivial for quantifiers, for illustration assume $\N\vDash_{\sigma} \dot\forall\phi$ for some environment $\sigma:\Nat\sto\Term$ and formula $\phi$.
	We aim to show $\M\vDash_{\hat\rho\circ\sigma} \dot\forall\phi$ which is equivalent to $\M\vDash_\rho \dot\forall\phi[\uparrow\!\sigma]$ and thus reduces to $\M\vDash_\rho\phi[\uparrow\!\sigma][\rho\,n]$ using a witness $\rho\,n$ guaranteed by the Henkin property of $\rho$.
	The latter then follows from $\N\vDash_{\sigma} \dot\forall\phi$ instantiated to $\rho\,n$ and the inductive hypothesis.
\end{proof}

All upcoming proofs of the $\LS$ theorem will factor through \Cref{thm_LS_environment} or a strengthening thereof (\Cref{thm_LS_blurred}).

\section{Downward Löwenheim-Skolem Theorem using DC and LEM}
\label{sec_LS_classical}

In this section, we use the proof strategy induced by \Cref{thm_LS_environment} to reconstruct the well-known connection of the $\LS$ theorem to $\DC$ over a classical meta-theory~\cite{espindolalowenheim,karagila2014downward,boolos2002computability}, providing both $\CC$ and $\LEM$.
First, we show that in this context, $\DC$ can be used to construct a Henkin environment and therefore to conclude the $\LS$ theorem.
As all the later and constructively refined proofs will follow the same pattern, we give the construction here in full detail.

\setCoqFilename{HenkinEnv}

\begin{theorem}[ ][LS_downward_with_DC_LEM]
	\label{thm_LS_LEM}
	Assuming $\DC+\LEM$, the $\LS$ theorem holds.
\end{theorem}

\begin{proof}
	By \Cref{thm_LS_environment}, it is enough to show that under the given assumptions every model admits a Henkin environment.
	Given a model $\M$, the construction of Henkin environment is done in three steps, each making use of a different logical assumption, thereby explaining the respective non-constructive contributions.
	The high-level idea is to describe an extension method by which Henkin witnesses are accumulated stage by stage, where $\LEM$ is needed to guarantee the existence of Henkin witnesses, $\CC$ (as a consequence of $\DC$) is needed to pick such witnesses simultaneously for every formula in every stage, and finally $\DC$ is needed to obtain a path through all stages such that its union constitutes a Henkin environment.
	
	Formally, we define a step relation $S:(\Nat\sto\M)\sto(\Nat\sto\M)\sto\Prop$ to express the extension of environments, such that $S\,\rho\,\rho'$ captures that $\rho'$ contains all witnesses with respect to $\rho$:
	$$S\,\rho\,\rho'~:=\rho\subseteq \rho'\land\forall \phi.\bigwedge\begin{array}{ll}\exists n.\,\M\vDash_\rho\phi[\rho'\,n]~\to~\M\vDash_\rho\dot\forall\phi\\[0.1cm] \exists n.\,\M\vDash_\rho\dot\exists\phi~\to~\M\vDash_\rho\phi[\rho'\,n]\end{array}$$
	Clearly every fixed point of $S$, i.e.~$\rho$ with $S\,\rho\,\rho$, is a Henkin environment so we now explain how such a fixed point is obtained by the aforementioned three steps.
	\begin{enumerate}
		\item
		Given any environment $\rho$, the assumption of $\LEM$ guarantees Henkin witnesses to exist for all formulas by its connection to the drinker paradoxes:
		For $\dot\forall \phi$, the existence of a Henkin witness is exactly the instance $\DP_\M$ for the predicate $\M\vDash_\rho\phi[\_]$ and for $\dot\exists \phi$ exactly the corresponding instance $\DDP_\M$.
		\item
		We now use $\CC_\M$ to show that $S$ is total, i.e.~given some $\rho$ we construct $\rho'$ with $S\,\rho\,\rho'$.
		By the previous step, we know that every formula $\dot\forall\phi$ has a Henkin witness with respect to $\rho$.
		So by fixing some enumeration $\phi_n$ of formulas, we know that for every $n$ the formula $\dot\forall\phi_n$ has a Henkin witness and thus $\CC_\M$ yields a function $\rho_\forall$ such that $\rho_\forall\,n$ is the Henkin witness to $\dot\forall\phi_n$.
		Analogously, another application of $\CC_\M$ yields a function $\rho_\exists$ such that $\rho_\exists\,n$ is the Henkin witness to $\dot\exists\phi_n$.
		We then set $\rho':=\rho\cup(\rho_\forall\cup\rho_\exists)$ and obtain $S\,\rho\,\rho'$ by simple calculation.
		\item
		We apply $\DC_{\Nat\sto\M}$ to get a path $F:\Nat\sto(\Nat\sto\M)$ through $S$, yielding a cumulative sequence of environments $F_0\subseteq F_1\subseteq F_2\subseteq\dots$ of Henkin witnesses.
		To collect the sequence into a single environment, we define
		$$\rho\,\langle n_1,n_2\rangle~:=~F_{n_1}\,n_2$$
		and verify that $S\,\rho\,\rho$, i.e.~that $\rho$ is Henkin.
		This is obtained by several properties of $\rho$:
		\begin{itemize}
			\item
			$F_k\subseteq \rho$ for every $k$:
			Given $n$ we need to find $n'$ with $F_k\,n=\rho\,n'$, pick $n':=\langle k,n\rangle$.
			\item
			$S\,F_k\,\rho$ for every $k$:
			By the previous fact, we know $F_k\subseteq \rho$, so we just need to show that $\rho$ contains all Henkin witnesses relative to $F_k$.
			Since $F$ is a path through $S$, we know $S\,F_k\,F_{k+1}$, so $F_{k+1}$ contains these witnesses, but then so does $\rho$ given $F_{k+1}\subseteq \rho$.
			\item
			$S\,\rho\,\rho$:
			Since $\rho\subseteq \rho$, we just need to show that for a given $\phi$ both Henkin witnesses relative to $\rho$ are contained in $\rho$.
			Since $\phi$ contains only finitely many variables and therefore, since $\rho$ is constructed in cumulative stages, we can find $k$ with $\rho\subseteq_\phi F_k$, meaning $\rho$ is included in $F_k$ on all free variables of $\phi$.
			Then in particular there is a permutation substitution $\sigma$ such that evaluation of $\phi$ in $\rho$ coincides with evaluation of $\phi[\sigma]$ in $F_k$.
			But then, since $S\,F_k\,\rho$ by the previous fact, $\rho$ contains the witnesses for $\phi[\sigma]$ relative to $F_k$ and thus for $\phi$ relative to $\rho$ itself.
			\qedhere
		\end{itemize}
	\end{enumerate}
\end{proof}

We remark that the forthcoming constructive refinements will weaken the respective logical assumptions in each of the three steps above, making precise which independent sources of non-constructivity are at play.
For the converse direction, the necessity for dependent choice relies on the presence of countable choice.

\setCoqFilename{ReverseLS}
\begin{coqfact}[ ][LS_CC_impl_DC]
	\label{fact_rev_cc}
	Assuming $\CC_\Nat$, the $\LS$ theorem implies $\DC$.
\end{coqfact}

\begin{proof}
	The high-level idea is that the $\LS$ theorem reduces $\DC_A$ to $\CC_\Nat$ by transforming $A$ into a countable domain.
	
	Formally, assuming a total relation $R:A\sto A\sto\Prop$, we consider the model $\M$ with domain $A$ and interpretation $P^\M_R\,x\,y:= R\,x\,y$ for some binary relation symbol $P_R$.
	The $\LS$ theorem then yields an elementary submodel $\N$ over a countable domain, say $\Nat$ itself for simplicity, witnessed by an elementary homomorphism $h:\N\sto\M$.
	Since totality is a first-order property with $\M\vDash\total{P_R}$ by assumption, in particular $\N\vDash \total{P_R}$, so the interpretation $P^\N_R:\Nat\sto\Nat\sto\Prop$ must be total, too.
	
	But then $\CC_\Nat$ yields a choice function $f:\Nat\sto\Nat$ for $P^\N_R$ and we can verify that $g:\Nat\sto A$ defined by $g\,n:=h\,(f^n\,0)$ is a path through $R$:
	to justify $R\,(g\,n)\,(g\,(n+1))$ for any $n$, consider an environment $\rho :\Nat\sto\N$ with $\rho\,0:=f^n\,0$ and $\rho\,1:=f^{n+1}\,0$, so $R\,(g\,n)\,(g\,(n+1))$ can be equivalently stated as $\M\vDash_{h\circ \rho} P_R(\x_0,\,\x_1)$.
	By elementarity of $h$ this reduces to $\N\vDash_{\rho} P_R(\x_0,\,\x_1)$, which translates to $P^\N_R\,(f^n\,0)\,(f\,(f^n\,0))$ and holds since $f$ is a choice function for $P^\N_R$.
\end{proof}

\setCoqFilename{AnalysisLS}
\begin{corollary}[Classical Decomposition][LS_iff_DC_under_CC_nat_LEM]
	\label{larry_LS_classical}
	Over $\CC_\Nat+\LEM$ as base assumptions, the $\LS$ theorem is equivalent to $\DC$.
\end{corollary}

All upcoming derivations of logical principles from the $\LS$ theorem will follow the same pattern of turning a given structure into a countable substructure, deriving a certain property in the simpler countable case, and reflecting it back to the original case.
While it seems impossible to derive the full strength of $\DC$ from the $\LS$ theorem, as the latter only reduces $\DC$ to the constructively still unprovable $\CC$, we observe that the restriction of $\DC$ to \emph{decidable} relations can be derived, as it then reduces $\CC$ to the provable principle $\CC^\Delta_\Nat$.

\setCoqFilename{ClassicalDC}
\begin{coqfact}[ ][LS_impl_DC_delta]
	\label{fact_delta}
	The $\LS$ theorem implies $\DC^\Delta$.
\end{coqfact}

\begin{proof}
	As in the proof of \Cref{fact_rev_cc} we obtain a total relation $P^\N_R:\Nat\sto\Nat\sto\Prop$ induced by the $\LS$ theorem for a model encoding a total relation $R:A\sto A\sto\Prop$.
	Now since we assume that $R$ is decidable, so is $P^\N_R$ by elementarity and then $\CC^\Delta_\Nat$ yields a choice function $f:\Nat\sto\Nat$ for $P^\N_R$.
	From there we proceed as before.
\end{proof}

Regarding the contribution of $\LEM$ in the form of the drinker paradoxes needed for the Henkin witnesses in each extension step, there is no chance to fully reverse the result:
For instance to derive $\DP_A$, we could start from a predicate $P:A\sto \Prop$ but even when using the $\LS$ theorem to reduce $P$ to a countable sub-predicate $P':\Nat\sto\Prop$, we have no means to find a particular $n$ such that $P'\,n$ would imply $\forall n.\,P'\,n$ and therefore $\forall x.\,P\,x$.
In other words, while the $\LS$ theorem reduces $\DP_A$ to $\DP_\Nat$, by \Cref{fact_lpo} we would still need at least $\LPO$ to proceed deriving $\DP_\Nat$.
Instead, in the next section we introduce weakenings of the drinker paradoxes that do become provable in the countable case while still being strong enough to derive the $\LS$ theorem.

\section{The Blurred Drinker Paradox}
\label{sec_BDP}

In this section, we introduce the concept of \emph{blurring}, by which we refer to replacing existential quantifiers by quantification over subsets.
By this transformation, logical principles can be obtained with constructively slightly reduced information content, as concrete witnesses are hidden in a blur of computationally indistinguishable elements.
Here, we study that concept at the example of the drinker paradoxes, in \Cref{sec_BAC} we will extend it to choice principles.
A summary diagram will be given at the end of this section.

We introduce the following blurred forms of $\DP$ and $\DDP$:
\begin{align*}
	\BDP^B_A~:=~&\forall P:A\sto\Prop.\exists f:B\sto A.\,(\forall y.\,P\,(f\,y))\to \forall x.\, P\,x\\
	\BDDP^B_A~:=~&\forall P:A\sto\Prop.\exists f:B\sto A.\,(\exists x.\, P\,x)\to\exists y.\,P\,(f\,y)
\end{align*}

Building on the intuition from before, for instance the principle $\BDP^B_A$ states that a Henkin witness for $P:A\sto\Prop$ in the sense of $\DP_A$ is contained in a blur of size at most $B$, represented by a function $f:B\sto A$.
In that situation, we call $f$ a \emph{blurred} Henkin witness or simply a \emph{Henkin blur} and require that $B$ is inhabited.

Note that, while $\DP_A$ and $\DDP_A$ are duals in the sense that $\DP_A$ also yields $\DDP_A$ for negative predicates $\{ x:A\mid \neg P\,x \}$ and vice versa, even in that sense $\BDDP^B_A$ is still slightly stronger than $\BDP^B_A$ as it concludes with a constructively strong existential quantifier.
This will play a role in the slightly different connection to Kripke's schema subject to \Cref{fact_bdp_ks}.

We first collect some properties of the introduced principles:

\setCoqFilename{LogicalPrinciples}
\begin{coqfact}[ ][scheme_facts_basic]
	\label{fact_bdp}
	The following statements hold, the latter two requiring that $B$ be inhabited:
	\begin{enumerate}
		\coqitem[scheme_facts_basic11]
		Both $\BDP^A_A$ and $\BDDP^A_A$.
		\coqitem[scheme_facts_basic2]
		If $\BDP^B_A$ and $\BDP^C_B$, then $\BDP^C_A$.
		\coqitem[scheme_facts_basic3]
		If $\BDDP^B_A$ and $\BDDP^C_B$, then $\BDDP^C_A$.
		\coqitem[scheme_facts_basic41]
		$\DP_A$ implies $\BDP^B_A$ and is equivalent to $\BDP^\Unit_A$.
		\coqitem[scheme_facts_basic51]
		$\DDP_A$ implies $\BDDP^B_A$ and is equivalent to $\BDDP^\Unit_A$.
	\end{enumerate}
\end{coqfact}

\begin{proof}
	We prove each claim independently.
	\begin{enumerate}
		\item
		By choosing $f$ to be the identity function.
		\item
		Assuming $P:A\sto \Prop$, given $f_1:B\sto A$ from $\BDP^B_A$ for $P$ and $f_2:C\sto B$ from $\BDP^C_B$ for $P\circ f_1$, the composition $f_1\circ f_2$ witnesses $\BDP^C_A$ for $P$.
		\item
		Analogous to (2).
		\item
		Assuming $P:A\sto \Prop$, $\DP_A$ for $P$ yields a Henkin witness $x$ for $P$ and the constant function $f\,y:=x$ then witnesses $\BDP^B_A$.
		Next, if $f:\Unit\sto A$ witnesses $\BDP^\Unit_A$ for $P$, then $f\,\star$ witnesses $\DP_A$ for $P$.
		\item
		Analogous to (4).
		\qedhere
	\end{enumerate}
\end{proof}

Note that by (1) in particular $\BDP^\Nat_\Nat$ and $\BDDP^\Nat_\Nat$ hold, meaning that in light of the concluding remark in \Cref{sec_LS_classical} we indeed face weakenings of the drinker paradoxes, provable in the countable case.
For simplicity, from now on we write $\BDP$ to denote $\BDP^\Nat_A$ for all inhabited $A$, as the case of countable blurring is the most relevant one, same for $\BDDP$.

To illustrate the generality of the blurring concept, we compare the blurred drinker paradox to a blurred form of $\IP$:
\begin{align*}
\BIP_A^B ~:=~ &\forall P:A\sto\Prop.\,\forall p:\Prop.\,(p\to \exists x.\,P\,x)\to \exists f:B\sto A.\,p\to \exists y.\,P\,(f\,y)
\end{align*}

For $\BIP$ we could show similar properties as in \Cref{fact_bdp}, stating that it is a generalisation of $\IP$ into a hierarchy of principles.
Instead, we generalise the equivalence of $\DDP$ and $\IP$ recorded in \Cref{fact_lpo}.

\begin{coqfact}[ ][BIP_iff_BEP]
	$\BDDP_A^B$ is equivalent to $\BIP_A^B$.
\end{coqfact}

\begin{proof}
	Analogous to the proof of \Cref{fact_lpo}, for the backward direction choose $p:=\exists x.\,P\,x$ as before.
\end{proof}

Intuitively, the blurred drinker paradoxes allow to test quantified properties on a large domain by considering restrictions to smaller domains, especially countable ones.
In this perspective, they resemble Kripke's schema~\cite{van1977use}, stating that every proposition can be tested by considering the solvability of Boolean functions over a countable domain:
\begin{align*}
	\KS~:=~&\forall p:\Prop.\exists f:\Nat\sto\Bool.\,p\leftrightarrow \exists n.\,f\,n=\btrue\\
	\KS'~:=~&\forall p:\Prop.\exists f:\Nat\sto\Bool.\,(p\sto \neg(\forall n.\,f\,n=\bfalse))\land((\exists n.\,f\,n=\btrue)\sto p)
\end{align*}

Note that $\KS$ expresses that every proposition is $\Sigma_1$, where the logical complexity class $\Sigma_1$ refers to the syntactic form of a single existential quantifier over a decidable predicate.
In comparison, the slightly weaker $\KS'$ replaces the existential quantifier in one direction by a negated universal quantifier.

We now establish the connection of the blurred drinker paradoxes to these formulations of Kripke's schema:

\begin{coqfact}[ ][scheme_facts_1]
	\label{fact_bdp_ks}
	$\BDP$ implies $\KS'$ and $\BDDP$ implies $\KS$.
\end{coqfact}

\begin{proof}
	We show that $\BDDP$ implies $\KS$, the other claim is similar.
	So for $p:\Prop$, consider $A:=\{b:\Bool\mid b=\bfalse \lor p\}$ and
	$$P\,b~:=~\begin{cases}p&\text{if }b=\btrue\\ \bot&\text{otherwise}\end{cases}$$
	for which $\BDDP^\Nat_A$ yields a Henkin blur $f:\Nat\sto A$.
	The induced underlying function $g:\Nat\sto\Bool$ then witnesses $\KS$ for $p$:
	First assuming $p$, we can show $\exists b.\,P\,b$ by using $b=\btrue$.
	Then by the Henkin property of $f$ we obtain $\exists n.\,P\,(f\,n)$ and thus $\exists n.\,g\,n=\btrue$.
	Conversely, if $g\,n=\btrue$ for some $n$, then by construction $p$ can be derived.
\end{proof}

Note that Kripke's schema can also be formulated for arbitrary $B$ in the role of $\Nat$, then admitting the same connections for drinker paradoxes blurred by $B$, see \Cref{sec_general}.
In that sense, the latter can be seen as a generalisation of Kripke's schema.

To further characterise the strength of the blurred drinker paradoxes, note that the difference between $\KS$ and $\KS'$ disappears in the presence of Markov's principle~\cite{Markov1953-MARTTO-31}, stating that $\Sigma_1$ propositions satisfy double negation elimination:
$$\MP ~:=~ \forall f:\Nat\sto\Bool.\,\neg\neg(\exists n.\, f\,n=\btrue)\to \exists n.\, f\,n=\btrue$$

It is straightforward to see that $\MP$ follows from $\LPO$ and thus from $\DP_\Nat$ by \Cref{fact_lpo}.
Since it is also well-known that $\MP$ together with $\KS$ and thus already with $\KS'$ implies $\LEM$, we obtain the following decompositions of $\LEM$ into blurred drinker paradoxes and side conditions.

\begin{coqfact}[ ][scheme_facts_2]
	\label{fact_BDP_LEM}
	The following are equivalent to $\LEM$:
	\vspace{-0.4cm}
	\begin{multicols}{2}
		\begin{enumerate}
			\coqitem[BDP_DP_nat_iff_LEM]
			$\BDP+\DP_\Nat$
			\coqitem[BDP_MP_iff_LEM]
			$\BDP+\MP$
			\coqitem[BEP_EP_nat_iff_LEM]
			$\BDDP+\DDP_\Nat$
			\coqitem[BEP_MP_iff_LEM]
			$\BDDP+\MP$
		\end{enumerate}
	\end{multicols}
\end{coqfact}

\begin{proof}
	That $\LEM$ implies (1)-(4) follows from previous observations.
	We show that (1) and (4) both imply $\LEM$, analogous arguments work for (2) and (3):
	\begin{itemize}
		\item
		By \Cref{fact_dp_lem} it is enough to show $\DP$, i.e.~$\DP_A$ for every inhabited $A$.
		By (4) of \Cref{fact_bdp}, this amounts to showing $\BDP^\Unit_A$, which decomposes into $\BDP^\Nat_A$ and $\BDP^\Unit_\Nat$ by (2) of \Cref{fact_bdp}.
		The former is an instance of $\BDP$ and the latter is equivalent to $\DP_\Nat$ by again using (4) of \Cref{fact_bdp}.
		\item
		By \Cref{fact_bdp_ks}, $\BDDP$ implies $\KS$ and the latter together with $\MP$ implies $\LEM$ by a standard argument:
		Given a proposition $p$, using $\KS$ for the claim $p\lor\neg p$ yields $f:\Nat\sto\Prop$ such that $p\lor\neg p$ is equivalent to $\exists n.\,f\,n=\btrue$.
		By $\MP$, it is enough to show $\neg\neg(\exists n.\,f\,n=\btrue)$ and hence $\neg\neg(p\lor\neg p)$, the latter being a tautology.
		\qedhere
	\end{itemize}
\end{proof}

We summarise the connections of the blurred drinker paradoxes with related principles in the following diagram:


\[\begin{tikzcd}[column sep=3.15em,row sep=large]
& {\LEM/ \DP/ \DDP/ \IP} \\
\BDP & \LPO & \BDDP \\
{\KS'} & \MP & \KS
\arrow[from=1-2, to=2-2]
\arrow[from=1-2, to=2-1]
\arrow[from=1-2, to=2-3]
\arrow[from=2-1, to=3-1]
\arrow[from=2-2, to=3-2]
\arrow[from=2-3, to=3-3]
\arrow[""{name=0, anchor=center, inner sep=0}, dashed, no head, from=3-2, to=3-3]
\arrow[""{name=1, anchor=center, inner sep=0}, dashed, no head, from=3-1, to=3-2]
\arrow[curve={height=12pt}, shorten <=9pt, dashed, from=0, to=1-2]
\arrow[curve={height=-12pt}, shorten <=9pt, dashed, from=1, to=1-2]
\end{tikzcd}\]

Here, the solid arrows depict (presumably strict) implications while the dashed arrows depict combined equivalences.

\section{Downward Löwenheim-Skolem Theorem using DC and BDP}
\label{sec_LS_blurred}

We now come back to the $\LS$ theorem and explain how the blurred drinker paradoxes from the previous section capture the contribution of classical logic below $\LEM$, postponing the orthogonal analysis of choice principles below $\DC$.
To this end, we first develop a strengthening of \Cref{thm_LS_environment} by observing that a weaker form of Henkin environments suffices to construct elementary submodels.

\setCoqFilename{ConstructiveLS}
\begin{definition}[Blurred Henkin Environment][Blurred_Henkin_env]
	Given a model $\M$, we call $\rho:\Nat\sto\M$ a \emph{blurred Henkin environment} if it collects Henkin witnesses for every formula $\phi$ as follows:
	\begin{align*}
	(\forall n.\,\M\vDash_\rho\phi[\rho\,n])~\to~\M\vDash_\rho\dot\forall\phi\\
	\M\vDash_\rho\dot\exists\phi~\to~(\exists n.\,\M\vDash_\rho\phi[\rho\,n])
	\end{align*}
\end{definition}

Note that every Henkin environment is a blurred Henkin environment, but not vice versa.
Still, the latter are enough to derive the $\LS$ theorem, as in the construction of the syntactic model actually no concrete witnesses are needed but just a guarantee that they are among the elements selected by the environment.

\begin{theorem}[DLS via Blurring][Blurred_Henkin_LS]
	\label{thm_LS_blurred}
	For every model admitting a blurred Henkin environment one can construct a syntactic elementary submodel.
\end{theorem}

\begin{proof}
	This is basically the same as \Cref{thm_LS_environment} where, for instance, in the critical direction of universal quantification we assume that the syntactic model $\N$ induced by $\rho$ satisfies $\N\vDash_{\sigma} \dot\forall\phi$ for some environment $\sigma:\Nat\sto\Term$ and formula $\phi$ and need to show $\M\vDash_{\hat\rho\circ\sigma} \dot\forall\phi$.
	The latter is equivalent to $\M\vDash_\rho \dot\forall\phi[\uparrow\!\sigma]$ and thus reduces to $\forall n.\,\M\vDash_\rho\phi[\uparrow\!\sigma][\rho\,n]$ using the Henkin property of $\rho$.
	For some given $n$, the claim follows from $\N\vDash_{\sigma} \dot\forall\phi$ instantiated to $\rho\,n$ and the inductive hypothesis.
\end{proof}

Following the structure of \Cref{thm_LS_LEM}, we now derive the $\LS$ theorem from \Cref{thm_LS_blurred} by iteratively constructing blurred Henkin environments.
The previous use of $\LEM$ is now replaced by $\BDP$ to accommodate universal quantification, and by $\BDDP$ to accommodate existential quantification.

\setCoqFilename{HenkinEnv}
\begin{theorem}[ ][LS_downward_with_BDP_BEP_DC]
	\label{thm_LS_blurred'}
	Assuming $\DC+\BDP+\BDDP$, the $\LS$ theorem holds.
\end{theorem}

\begin{proof}
	We employ \Cref{thm_LS_blurred}, leaving us with the construction of a blurred Henkin environment for an arbitrary model $\M$.
	This construction follows the same outline as in the proof of \Cref{thm_LS_LEM}, i.e.~we devise a step relation $S$ accumulating Henkin witnesses and obtain a blurred Henkin environment as a fixed point of $S$ in three steps.
	As step relation $S\,\rho\,\rho'$, we this time only require that $\rho'$ is a Henkin blur for all formulas $\phi$ relative to $\rho$, instead of the stronger requirement to provide concrete witnesses:
	$$S\,\rho\,\rho'~:=\rho\subseteq \rho'\land\forall \phi.\bigwedge\begin{array}{ll}(\forall n.\,\M\vDash_\rho\phi[\rho'\,n])~\to~\M\vDash_\rho\dot\forall\phi\\[0.1cm] \M\vDash_\rho\dot\exists\phi~\to~(\exists n.\,\M\vDash_\rho\phi[\rho'\,n])\end{array}$$
	\begin{enumerate}
		\item
		Given $\rho$ and $\phi$ there is a guarantee to be able to proceed, as the instance $\BDP_\M$ for the predicate $\M\vDash_\rho\phi[\_]$ yields a Henkin blur for $\dot\forall\phi$ and the same instance of $\BDDP_\M$ a Henkin blur for $\dot\exists\phi$.
		\item
		We derive totality of $S$ at $\rho$ using $\CC_{\Nat\sto\M}$ (following from $\DC$) on the previous fact, thus yielding choice functions $f_\forall,f_\exists:\Nat\sto(\Nat\sto\M)$ such that $f_\forall \,n$ is a Henkin blur for $\dot\forall\phi_n$ and $f_\exists \,n$ is a Henkin blur for $\dot\exists\phi_n$.
		By using Cantor pairing again, they induce environments $\rho_\forall\,\langle n1,n2\rangle:=f_\forall\,n_1\,n_2$ and $\rho_\exists\,\langle n1,n2\rangle:=f_\exists\,n_1\,n_2$ and for the choice $\rho':=\rho\cup(\rho_\forall\cup \rho_\exists)$ it is straightforward to verify $S\,\rho\,\rho'$ as desired.
		\item
		Finally, we can use $\DC_{\Nat\sto\M}$ to obtain a path $F:\Nat\sto(\Nat\sto\M)$ through $S$ and verify that $\rho\,\langle n1,n2\rangle:=F_{n_1}\,n_2$ is a fixed point of $S$ and thus a blurred Henkin environment similarly as before.
		\qedhere
	\end{enumerate}
\end{proof}

Note that restricting to the negative fragment of FOL, only $\BDP$ would be needed, meaning the non-constructive contributions of both sorts of quantification in the $\LS$ theorem are independent.
Conversely, from the $\LS$ theorem over the negative fragment we can derive $\BDP$, and with existential quantification present, also $\BDDP$ becomes derivable.

\setCoqFilename{ReverseLS}
\begin{coqfact}[ ][LS_impl_BEP]
	\label{fact_bdp_reverse}
	The $\LS$ theorem implies $\BDP+\BDDP$.
\end{coqfact}

\begin{proof}
	We show how to derive $\BDP$ from the $\LS$ theorem, the case of $\BDDP$ is dual.
	Similar to the reverse proofs given in \Cref{sec_LS_classical}, the high-level idea is that the $\LS$ theorem reduces $\BDP_A^\Nat$ to the provable $\BDP_\Nat^\Nat$.
	
	Formally, assume a predicate $P:A\sto\Prop$ for some inhabited $A$, which we encode as a model $\M$ over $A$ by $P^\M\,x:=P\,x$.
	Then there must be an elementary embedding $h:\N\sto\M$ from some countable model $\N$, conceived over the domain $\Nat$ for simplicity.
	
	Since in $\N$ we do have a function $f:\Nat\sto\Nat$ such that $\forall n.\, P^\N\,(f\,n)$ implies $\forall n.\, P^\N\,n$, for instance by taking $f$ to be the identity, we obtain that $h\circ f$ is a Henkin blur for $P$ as follows:
	Assuming $\forall n.\, P\,(h\,(f\,n))$ we show $\forall n.\, P^\N\,(f\,n)$ by fixing $n$ and formulating $P^\N\,(f\,n)$ as $\N\vDash_\rho P(\x_0)$ for $\rho\,0:=f\,n$, which by elementarity follows from $\M\vDash_{h\circ\rho} P(\x_0)$, that is the assumption $P\,(h\,(f\,n))$.
	But then $\forall x.\, P^\N\,x$, which again reflects up into $\M$ using $h$ and thus yields $\forall x.\, P\,x$.
\end{proof}

\setCoqFilename{AnalysisLS}
\begin{corollary}[Blurred Decomposition][LS_iff_DC_BDP_BEP_under_CC_nat]
	\label{larry_blurred}
	Over $\CC_\Nat$ as base assumption, the $\LS$ theorem is equivalent to \mbox{$\DC+\BDP+\BDDP$}.
\end{corollary}

That means, disregarding the orthogonal contribution of choice principles, the logical strength of the $\LS$ theorem corresponds exactly to the blurred drinker paradoxes.

\section{Blurred Choice Axioms}
\label{sec_BAC}

In order to complete the analysis, in this section we discuss similarly blurred forms of choice principles that allow a precise decomposition of the $\LS$ theorem.
For simplicity, we will consider the concrete case of countable blurring, i.e.~using functions $f:\Nat\sto A$ but sketch more general formulations at a later point~(\Cref{sec_future}).
Again, a summary diagram will be given at the end of this section.

We begin with a blurring of countable choice that weakens the information provided by a choice function for a total relation by hiding the choices within a countable subset:
{\small
$$\BCC_A~:=~\forall R:\Nat\sto A\sto \Prop.\,\total R \to \exists f:\Nat\sto A.\forall n.\exists m.\,R\,n\,(f\,m)$$
}

As usual, we write $\BCC$ to denote $\BCC_A$ for all $A$, similarly for all upcoming choice principles.
In the situation of $\BCC_A$ we call $f:\Nat\sto A$ a \emph{blurred} choice function.
Note that in the case of $A:=\Nat$ the identity on $\Nat$ is a blurred choice function, so as in the case of the blurred drinker paradoxes we have the desired property that $\BCC$ and all upcoming blurred choice principles hold in the countable case, suggesting their connection to the $\LS$ theorem.
Moreover, blurred choice principles follow from their regular counterparts, allowing the following decomposition of countable choice:

\setCoqFilename{LogicalPrinciples}
\begin{coqfact}[ ][CC_impl_BCC_on]
	\label{fact_bcc}
	$\CC$ is equivalent to $\BCC+\CC_\Nat$.
\end{coqfact}

\begin{proof}
	To show that $\CC_A$ implies $\BCC_A$, for a total relation $R:\Nat\sto A\sto \Prop$ we obtain a choice function $f:\Nat\sto A$ which in particular can be considered a blurred choice function.
	
	Starting from $\BCC_A$, an application of $\CC_\Nat$ is enough to turn a blurred choice function into a choice function.
\end{proof}

We will see in \Cref{sec_full} that $\BCC$ is enough to handle step (2) of the construction in \Cref{thm_LS_blurred'}, i.e.~to derive totality of the step relation $S$.
Regarding step (3), i.e.~the derivation of a fixed point for $S$, we need to find a weakening of $\DC$ without the contribution of $\CC$, so that it becomes provable in the countable case.
A first attempt is as follows, where we simply replace the path through a total relation $R$ guaranteed by $\DC$ by a countable and total sub-relation:
$$\BDC_A~:=~\forall R:A\sto A\sto\Prop.\, \total R \to \exists f:\Nat\sto A.\,\total {R\circ f}$$

Note that by $R\circ f$ we refer to the pointwise composition of $R$ and $f$, i.e.~to the relation $(R\circ f)\,n\,m:=R\,(f\,n)\,(f\,m)$.
The obtained function $f$ is called a \emph{blurred} path as it still represents a sequence through $R$ but hides the respective successors.

We then show that, while implying $\BCC$, the obtained $\BDC$ needs some contribution of $\CC$ to get back the strength of $\DC$.

\begin{coqfact}[ ][Result]
	\label{fact_bdc}
	The following statements hold:
	\begin{enumerate}
		\coqitem[DC_impl_BDC_on]
		$\DC_A$ implies $\BDC_A$.
		\coqitem[BDC_impl_BCC]
		$\BDC$ implies $\BCC$.
		\coqitem[DC_iff_BDC_CC_nat]
		$\DC$ is equivalent to $\BDC+\CC_\Nat$.
	\end{enumerate}
\end{coqfact}

\begin{proof}
	We prove all claims independently:
	\begin{enumerate}
		\item
		Again as in \Cref{fact_bcc}, the blurred conclusion of $\BDC_A$ is visibly a weakening of the conclusion of $\DC_A$.
		
		\item
		First as in \Cref{fact_ac}, note that $\BDC$ can be equivalently stated for arbitrary $x_0:A$ as
		$$\forall R:A\sto A \sto\Prop.\,\total R\to \exists f.\,f\,0=x_0\land\total {R\circ f}$$
		by restricting $R$ to the sub-relation $R'$ reachable from $x_0$.
		Then a blurred path $f$ through $R$ induces a blurred path $f'$ through $R'$ by first taking the path from $x_0$ to $f\,0$ and by then continuing with $f$.
		
		Now to show $\BCC$, assume a total relation $R: \Nat\sto A\sto\Prop$ on $A$ with some $a_0$ and consider $A':=\Nat\times A$ and
		$$R'\,(n,x)\,(m,y)~:=~m = n+1\land R\,n\,y$$
		which is total since $R$ is total.
		The modified version of $\BDC$ for $R'$ and the choice $x_0:=(0, a_0)$ then yields a blurred path $f':\Nat\to\Nat\times  A$ through $R'$ and it remains to verify that $f\,n:=\pi_2\,(f'\,n)$ is a blurred choice function for $R$.
		
		First, using the properties of $f'$ we derive
		$$\forall n.\,\exists m.\,\pi_1 (f'\,m) = n$$
		by induction on $n$, choosing $0$ in the base case and, in the inductive step where we have some $m$ with $\pi_1 (f'\,m) = n$, by choosing $m'$ with $R'\,(f'\,m)\,(f'\,m')$ which we obtain by totality of $R'\circ f'$.
		
		Now, given some $n$, we find $m$ with $R\,n\,(f\,m)$ by first finding $m_1$ with $\pi_1 (f'\,m_1) = n$ as above and subsequently by finding $m_2$ with $R'\,(f'\,m_1)\,(f'\,m_2)$ via totality of $R'\circ f'$.
		Then $R\,n\,(f\,m_2)$ as this is equivalent to $R\,(\pi_1\,(f'\,m_1))\,(\pi_2\,(f'\,m_2))$ which in turn follows from $R'\,(f'\,m_1)\,(f'\,m_2)$.
		
		\item
		Given (1) and \Cref{fact_ac} it only remains to show that $\BDC$ and $\CC_\Nat$ together imply full $\DC$.
		So assume some total $R:A\sto A\sto\Prop$, then $\BDC$ yields $f:\Nat\sto A$ such that $R\circ f$ is total.
		The latter is a relation $\Nat\sto\Nat\sto\Prop$ to which $\CC_\Nat$ yields a choice function $g:\Nat\to\Nat$.
		A path $h:\Nat\to A$ through $R$ is then obtained by the function $h\,n:=f (g^n\,0)$.
		\qedhere
	\end{enumerate}
\end{proof}

Although $\BDC$ therefore yields the desired decomposition of $\DC$, it does not seem strong enough for the $\LS$ theorem.
Intuitively, the problem is that $\BDC$ does not have access to the history of previous choices that is needed to merge the environments in proof step (3) of \Cref{thm_LS_blurred'}.
This problem can be fixed by strengthening to relations on finite sequences $A^*$ or, sufficiently, over pairs $A^2$:
$$\BDC^2_A~:=~\forall R:A^2\sto A\sto\Prop.\, \total R \to \exists f:\Nat\sto A.\,\total {R\circ f}$$

As for $\BDC$, by $R\circ f$ we refer to pointwise composition of $R$ and $f$, this time component wise for pairs.
First note that $\BDC^2$ is indeed a strengthening of $\BDC$:

\begin{coqfact}[ ][BDC2_impl_BDC_on]
	\label{fact_bdc2}
	$\BDC^2_A$ implies $\BDC_A$.
\end{coqfact}

\begin{proof}
	Straightforward by turning $R:A\sto A\sto\Prop$ to show $\BDC_A$ into $R'\,(x,y)\,z:=R\,x\,z$ and then applying $\BDC^2_A$.
\end{proof}

We leave the fact that $\BDC^2$ also corresponds to a version of $\DC$ without the contribution of $\CC$ to a later point, as this proof will be indirect requiring intermediate structure, see \Cref{larry_bdc2}.

As we will see in \Cref{sec_full}, the principle $\BDC^2$ is already strong enough for the desired purpose of replacing $\DC$ in the proof of \Cref{thm_LS_blurred'}.
Moreover, it is possible to again weaken $\BDC^2$ to not even derive $\BDC$, thus completely orthogonalising the different ingredients for the $\LS$ theorem:
$$\DDC_A~:=~\forall R: A\sto A \sto \Prop.\, \directed R \to \exists f:\Nat\sto A.\, \directed{R\circ f}$$

Here, by $\directed R$ we refer to $R$ being \emph{directed}, i.e.~satisfying for every $x,y:A$ that there is $z:A$ with $R\,x\,z$ and $R\,y\,z$.
So informally, $\DDC$ states that every directed relation has a countable directed sub-relation, which captures the same idea leading to $\BDC^2$ that the information of two previous environments should be combinable.

Indeed, $\BDC^2$ can be decomposed independently into $\DDC$ and $\BCC$, with one direction akin to the iterative construction underlying \Cref{thm_LS_blurred'} and the forthcoming \Cref{thm_full}.

\begin{coqfact}[ ][Result]
	\label{fact_ddc}
	The following statements hold:
	\begin{enumerate}
		\coqitem[BDC2_impl_DDC]
		$\BDC^2_A$ implies $\DDC_A$.
		\coqitem[BDC2_iff_DDC_BCC]
		$\BDC^2$ is equivalent to $\DDC+\BCC$.
	\end{enumerate}
\end{coqfact}

\begin{proof}
	We prove both claims independently:
	\begin{enumerate}
		\item
		Directedness of $R: A\sto A \sto \Prop$ induces totality of
		$$R'\,(x,y)\,z~:=~ R\,x\,z \land R\,y\,z$$
		and, conversely, totality of a countable sub-relation $R'\circ f$ induces directedness of $R\circ f$.
		The claim follows.
		
		\item 
		The first direction follows from (1) and \Cref{fact_bdc,fact_bdc2}.
		For the converse, assume a total relation $R:A^2\sto A\sto\Prop$.
		Consider $S:(\Nat\sto A)\sto(\Nat\sto A)\sto\Prop$ defined by
		$$S\,\rho\,\rho'~:=~\rho\subseteq \rho' \land \forall n m.\exists k.\,R\,(\rho\,m, \rho\,n)\, (\rho'\,k)$$
		which can be shown total using $\BCC$ as follows:
		Given some $\rho$, consider the relation $R':\Nat\sto A\sto\Prop$ defined by
		$$R'\,\langle n_1,n_2\rangle\,x~:=~ R\,(\rho\,n_1, \rho\,n_2)\,x$$
		which is total since $R$ is total.
		Then $\BCC_A$ yields a blurred choice function $\rho':\Nat\sto A$ for $R'$ and it is easy to verify that $S\,\rho\,(\rho\cup \rho')$ holds, thus establishing totality of $S$ as desired.
		
		Employing totality, we obtain that $S$ is directed:
		Given $\rho_1$ and $\rho_2$ totality yields $\rho_1'$ and $\rho_2'$ with both $S\,\rho_1\,\rho_1'$ as well as $S\,\rho_2\,\rho_2'$.
		It then follows that both $S\,\rho_1\,(\rho_1'\cup\rho_2')$ and  $S\,\rho_2\,(\rho_1'\cup\rho_2')$ by simple calculation.
		
		We now apply $\DDC_{\Nat\sto A}$ to $S$ and obtain $F:\Nat\sto(\Nat\sto A)$ such that $S\circ F$ is directed.
		Then $\rho:\Nat\sto A$ defined by
		$$\rho\,\langle n_1,n_2\rangle~:=~F_{n_1}\,n_2$$
		can be shown to witness $\BDC^2$ for $R$ as desired:
		Indeed, to verify that $R\circ \rho$ is total (in fact stating that $\rho$ is a fixed point of $S$), we assume $n=\langle n_1,n_2\rangle$ and $m=\langle m_1,m_2\rangle$ and need to find $k$ with $R\,(\rho\,n)\,(\rho\,m)\,(\rho\,k)$.
		Using the directedness of $S\circ F$ for $n_1$ and $m_1$, we obtain $w$ with $F_{n_1}\subseteq F_w$ and $F_{m_1}\subseteq F_w$, so there are $n_3$ and $m_3$ with $F_{n_1}\,n_2=F_w\,n_3$ and $F_{m_1}\,m_2=F_w\,m_3$.
		Moreover, by totality of $S\circ F$ for $w$ we obtain $k_1,k_2$ with $R\,(F_w\,n_3)\,(F_w\,m_3)\,(F_{k_1}\,k_2)$ and thus $R\,(\rho\,n)\,(\rho\,m)\,(\rho\,k)$ for the choice $k:=\langle k_1,k_2\rangle$.
		\qedhere
	\end{enumerate}
\end{proof}

This decomposition of $\BDC^2$ into $\DDC$ and $\BCC$ then in particular entails the decomposition of $\DC$ into $\BDC^2$ and $\CC$.

\begin{coqfact}[ ][DC_impl_BDC2]
	$\DC$ implies $\BDC^2$.
\end{coqfact}

\begin{proof}
	We first show that $\DC_A$ implies a weaker version of $\DDC_A$ where the directed relation $R:A\sto A\sto \Prop$ is additionally required to be transitive.
	In that case and since directed relations are total, $\DC_A$ yields a path $f:\Nat\to A$ through $R$.
	It then remains to show that $R\circ f$ is directed, which follows since given w.l.o.g.\ $n<m$ we have both $R\,(f\,n)\,(f (m+1))$ using transitivity of $R$ along the path $f$ connecting $n$ and $m$, as well as $R\,(f\,m)\,(f (m+1))$ by a single step along $f$.
	
	Now since the relation $S$ defined in the proof part (2) of \Cref{fact_ddc} is transitive by construction, this modified version of $\DDC$ together with $\BCC$, following from $\DC$ by \Cref{fact_bdc}, is enough to derive $\BDC^2$ as before.
\end{proof}

\begin{corollary}[ ]
	\label{larry_bdc2}
	The following statements hold:
	\begin{enumerate}
		\coqitem[DC_iff_BDC2_CC_nat]
		$\DC$ is equivalent to $\BDC^2+\CC_\Nat$.
		\coqitem[DC_iff_DDC_CC]
		$\DC$ is equivalent to $\DDC +\CC$.
	\end{enumerate}
\end{corollary}

Finally we show that, similar to \Cref{fact_oac}, $\BDC^2$ also has an omniscient version that exactly adds $\BDP$ and $\BDDP$:
$$\OBDC^2_A~:=~\forall R:A^2\sto A\sto\Prop. \exists f:\Nat\sto A.\,\total R \leftrightarrow\total {R\circ f}$$

We here state only one direction of the decomposition for $\OBDC^2$ as the other direction follows more directly as a by-product of the full analysis of the $\DLS$ theorem in \Cref{sec_full}.

\begin{coqfact}[ ][OBDC_implies_BDP_BEP_BDC2]
	\label{fact_obdc}
	$\OBDC^2_A$ implies $\BDC^2_A+\BDP_A+\BDDP_A$.
\end{coqfact}

\begin{proof}
	We establish each claim separately:
	\begin{itemize}
		\item
		That $\OBDC^2_A$ implies $\BDC^2_A$ is as in \Cref{fact_oac}.
		\item
		To derive $\BDP_A$, assume $P:A\sto\Prop$ and set
		$$ R\,(x,y)\,z~:=~P\,x$$
		for which $\OBDC^2_A$ yields $f:\Nat\sto A$ such that $R$ is total if and only if $R\circ f $ is total, reducing to $P\,x$ for all $x$ if and only if $P\,(f\,n)$ for all $n$.
		So $f$ also witnesses $\BDP^2_A$.
		\item
		To similarly derive $\BDDP_A$, assume $P:A\sto\Prop$ and set
		$$ R\,(x,y)\,z~:=~P\,z$$
		because then any $f$ such that $R$ is total iff $R\circ f $ is total actually yields $P\,x$ for some $x$ iff $P\,(f\,n)$ for some $n$.
		\qedhere
	\end{itemize}
\end{proof}

We summarise the connections of the blurred choice axioms with related principles in the following diagram:


\[\begin{tikzcd}[column sep=3.15em,row sep=large]
\DC && \CC \\
{\BDC^2} & \BDC & \BCC \\
\DDC
\arrow[from=1-1, to=2-1]
\arrow[from=2-1, to=3-1]
\arrow[from=1-1, to=1-3]
\arrow[from=1-3, to=2-3]
\arrow[from=2-1, to=2-2]
\arrow[from=2-2, to=2-3]
\arrow[""{name=0, anchor=center, inner sep=0}, dashed, no head, from=3-1, to=2-3]
\arrow[""{name=1, anchor=center, inner sep=0}, dashed, no head, from=2-1, to=1-3]
\arrow[shorten <=5pt, dashed, from=0, to=2-1]
\arrow[shorten <=6pt, dashed, from=1, to=1-1]
\end{tikzcd}\]

As with the diagram at the end of \Cref{sec_BDP}, the solid arrows depict (presumably strict) implications while the dashed arrows depict combined equivalences.

\section{Full Analysis of the Downward Löwenheim-Skolem Theorem}
\label{sec_full}

In this section, we conclude the final decomposition of the $\LS$ theorem into the independent logical principles at play and combinations thereof.

\setCoqFilename{AnalysisLS}
\begin{theorem}[Decomposition][LSiffDC]
	\label{thm_full}
	The following are equivalent:
	\begin{enumerate}
		\coqitem[Decomposition1]
		The $\LS$ theorem
		\coqitem[Decomposition2]
		The conjunction of $\DDC$, $\BCC$, $\BDP$, and $\BDDP$
		\coqitem[Decomposition3]
		The conjunction of $\BDC^2$, $\BDP$, and $\BDDP$
		\coqitem[Decomposition4]
		The principle $\OBDC^2$
	\end{enumerate}
\end{theorem}

\begin{proof}
	We establish a circle of implications:
	\begin{itemize}
		\item
		That (4) implies (3) is by \Cref{fact_obdc}.
		\item
		That (3) implies (2) is by (2) of \Cref{fact_ddc}.
		\item
		That (2) implies (1) is a further refinement of \Cref{thm_LS_blurred'}.
		Again using \Cref{thm_LS_blurred}, we demonstrate how a blurred Henkin environment for any model $\M$ can be obtained as a fixed point of the step function $S$ from before:
		$$S\,\rho\,\rho'~:=\rho\subseteq \rho'\land\forall \phi.\bigwedge\begin{array}{ll}(\forall n.\,\M\vDash_\rho\phi[\rho'\,n])~\to~\M\vDash_\rho\dot\forall\phi\\[0.1cm] \M\vDash_\rho\dot\exists\phi~\to~(\exists n.\,\M\vDash_\rho\phi[\rho'\,n])\end{array}$$
		\begin{enumerate}
			\item
			As before, $\BDP_\M$ and $\BDDP_\M$ yield Henkin blurs $\rho'$ for every $\phi$ and $\rho$.
			\item
			For totality of $S$, this time using $\BCC_{\Nat\sto\M}$ instead of $\CC_{\Nat\sto\M}$ yields blurred choice functions $f_\forall,f_\exists:\Nat\sto(\Nat\sto\M)$, i.e.~we do not have that $f_\forall\,n$ is a Henkin blur for $\dot\forall \phi_n$ but only know that we can obtain such a Henkin blur by $f_\forall\,m$ for some $m$.
			Yet we can still easily verify that for $\rho_\forall$ and $\rho_\exists$ defined by pairing as before and the choice $\rho':=\rho\cup(\rho_\forall\cup\rho_\exists)$ we have that $S\,\rho\,\rho'$.
			\item
			To obtain a fixed point of $S$ using $\DDC_{\Nat\sto\M}$ instead of $\DC_{\Nat\sto\M}$, we first need to argue that $S$ is directed, which given $\rho_1$ and $\rho_2$ is easily done by using totality on $\rho_1\cup\rho_2$.
			Then from $\DDC_{\Nat\sto\M}$ we obtain $F:\Nat\sto(\Nat\sto\M)$ such that $S\circ F$ is directed and verify that the now familiar choice $\rho\,\langle n1,n2\rangle:=F_{n_1}\,n_2$ is a fixed point of $S$ and thus a blurred Henkin environment:
			the proofs that $Fk\subseteq \rho$ and $S\,F_k\,\rho$ are as before and to conclude $S\,\rho\,\rho$, we now use the directedness of $S\circ F$ to show that for every formula $\phi$ there is $k$ large enough such that $F_k$ already is a Henkin blur for $\rho$.
			For the latter, it is again enough to find $k$ with $\rho\subseteq_\phi F_k$, which is obtained by directedness for the finitely many $F_i$ contributing to the behaviour of $\rho$ on $\phi$.
		\end{enumerate}
	\item 
	That (1) implies (4) follows the same pattern as all reverse proofs from before, using that $\OBDC^2_\Nat$ is provable.
	Assuming $R:A^2\sto A\sto\Prop$ on inhabited $A$ taken as model $\M$, from the $\LS$ theorem we obtain an elementary embedding $h:\N\sto\M$ for a model $\N$ over domain $\Nat$.
	For the interpretation $R^\N$, e.g.~the identity function $f:\Nat\sto\Nat$ satisfies $\total {R^\N}$ iff $\total{R^\N\circ f}$.
	But then by elementarity also $h\circ f$ has that property, i.e.~$\total R$ iff $\total{R\circ (h\circ f)}$ can be derived as desired.
	\qedhere
	\end{itemize}
\end{proof}

Note that all of $\BDC^2$, $\DDC$, and $\BCC$ are also directly implied by the $\LS$ theorem, all following the same pattern as the derivation of $\BDP$ and $\BDDP$ already presented in \Cref{fact_bdp_reverse}.

\section{Generalised Blurred Axioms}
\label{sec_general}

\setCoqFilename{GeneralisedAxioms}

In this section, we study further properties of blurred axioms below $\LEM$ and $\AC$.
On the side of $\LEM$, those are versions of the blurred drinker paradox as well as matching generalisations of Kripke's schema and Markov's principle.
On the side of $\AC$, those are the natural generalisations of $\BDC$ and $\DDC$ for arbitrary blurring types as well as a generalisation of $\BCC$ for arbitrary relation domains.
Ultimately, we expect that these correspond to a suitable generalisation of $\DLS$ to uncountable signatures but only report on preliminary results in that direction.

We start with an observation concerning $\BDP$ for finite blurring types:

\begin{coqfact}[ ][BDP_LEM]
	\label{fact_bdp_bool}
	Already $\BDP^\Bool$ implies $\LEM$.
\end{coqfact}

\begin{proof}
	This in principle follows the same pattern as \Cref{fact_dp_lem} but we give a more type-theoretic version of the proof here for illustration purposes.
	So assuming $\BDP^\Bool$ and a proposition $p:\Prop$, consider the (non-empty) type $A:=p+\Unit$ and the predicate $P:A\to\Prop$ defined by
	\begin{align*}
		P\,(i_1\,H)&:= \bot\\
		P\,(i_2\,\star)&:= \top
	\end{align*}
	for which $\BDP^\Bool$ then yields a Henkin blur $f:\Bool\to A$.
	Now if one of $f\,\btrue$ or $f\,\bfalse$ has the form $i_1\,H$, in either case we obtain a proof $H:p$ and therefore $p$ holds.
	Otherwise, we know that $P\,(f\,b)$ for all $b:\Bool$ and we show $\neg p$ by assuming $p$ and deriving a contradiction.
	By the former and the Henkin property of $f$, we obtain $P\,x$ for all $x:A$ but then by the assumption of $p$ indeed we have some $H:p$ and therefore $P\,(i_1\,H)$, contradiction.
\end{proof}

In fact this observation generalises to arbitrary finite types, as long as the notion of finiteness admits an exhaustive case distinction as in the proof above.

Next, we formally record the overlap between $\BDP$ and $\BDDP$ as mentioned in \Cref{sec_BDP}:

\begin{coqfact}[ ][BDP_BEP]
	If $\BDDP^B_A$ holds for some predicate $P:A\to \Prop$, then $\BDP^B_A$ holds for $\overline P\,x :=\neg P\,x$.
\end{coqfact}

\begin{proof}
	Assume $\BDDP^B_A$ for $P:A\to \Prop$, so there is a Henkin blur $f:B\to A$ for $P$ in the sense of $\BDDP$.
	We show that it is also a Henkin blur for $\overline P$ in the sense of $\BDP$, so assuming that $\neg P\,(f\,y)$ for all $y:B$, we want to show that $\neg P\,x$ for all $x:A$.
	So let $P\,x$ for a contradiction, then the Henkin property for $P$ yields that there is $y:B$ with $P\,(f\,y)$, contradicting the above assumption.
\end{proof}

It seems implausible that a dual fact holds in the converse direction, given that $\BDDP$ ends in an existential quantifier.
Therefore the existing symmetry of $\DP$ and $\DDP$ is broken by moving to the blurred variants.

We continue with generalised versions of Kripke's schema anticipated in \Cref{sec_BDP}:
\begin{align*}
	\GKS_A~:=~&\forall p:\Prop.\exists f:A\sto\Bool.\,p\leftrightarrow \exists x.\,f\,x=\btrue\\
	\GKS'_A~:=~&\forall p:\Prop.\exists f:A\sto\Bool.\,(p\sto \neg(\forall x.\,f\,x=\bfalse))\land((\exists x.\,f\,x=\btrue)\sto p)
\end{align*}

Note that we recover $\KS$ as $\GKS_\Nat$ and $\KS'$ as $\GKS'_\Nat$ and that, intuitively, the generalised principles decrease in strength with increasing cardinality of $A$.
In the corner case of $A:=\Unit$ (and all other finite types), they have the full strength of $\LEM$:

\begin{coqfact}[ ][GKS_LEM]
	Both $\GKS_\Unit$ and $\GKS'_\Unit$ imply $\LEM$.
\end{coqfact}

\begin{proof}
	For the former claim, assuming $\GKS_\Unit$ and a proposition $p:\Prop$, we obtain a function $f:\Unit \to \Bool$ reflecting $p$ in the sense as above.
	Then by case analysis on $f\,\star$ we can decide whether $p$ or $\neg p$ holds.
	The latter claim is analogous.
\end{proof}

Note that from this perspective, $\KS$ can be seen as a countable blurring of $\LEM$.
Moreover, following the same proof pattern as in \Cref{fact_bdp_ks}, the way how blurred drinker paradoxes imply Kripke's schema generalises directly.

\begin{coqfact}[ ][BDP_GKS]
	\label{fact_BDP_GKS}
	$\BDP^B$ implies $\GKS'_B$ and $\BDDP^B$ implies $\GKS_B$.
\end{coqfact}

\begin{proof}
	Both are exactly as in \Cref{fact_bdp_ks}, with the role of $\Nat$ replaced by $B$.
	For illustration purposes, we discuss the latter case again in a more type-theoretical phrasing.
	So assume $\BDDP^B$ and some proposition $p:\Prop$.
	Consider the (non-empty) type $A:=p+\Unit$ and the predicate $P:A\to \Prop$ defined by
	\begin{align*}
		P\,(i_1\,H)&:= \top\\
		P\,(i_2\,\star)&:= \bot
	\end{align*}
	and obtain a respective Henkin blur $f:B\to A$ from $\BDDP^B$.
	Note that $f$ induces a function $g:B\to \Bool$ by setting $g\,y:=\btrue$ if $f\,y=i_1\,H$ and $g\,y:=\bfalse$ if $f\,y=i_2\,\star$.
	
	Then we show that $g$ witnesses $\GKS_B$ as follows.
	If $p$, so there is $H:p$, then $P\,(i_1\,H)$ holds so the property of $f$ being a Henkin blur yields some $y:B$ such that $P\,(f\,y)$.
	Now if $f\,y=i_1\,H'$ for some other proof of $p$, then by definition $g\,y=\btrue$ as needed for $\GKS_B$ and if $f\,y=i_2\,\star$, then $P\,(f\,y)$ is a contradiction.
	Conversely, if there is $y$ with $g\,y=\btrue$, then by definition $f\,y=i_1\,H$ for some $H:p$ and therefore $p$ holds.
\end{proof}

The matching generalisation of Markov's principle can be given as follows:
$$\GMP_A ~:=~ \forall f:A\sto\Bool.\,\neg\neg(\exists x.\, f\,x=\btrue)\to \exists x.\, f\,x=\btrue$$
Here we recover $\MP$ as $\GMP_\Nat$ and observe that, intuitively, the generalised principle increases in strength with increasing cardinality of $A$.
In the corner case of $A:=\Unit$ (and all other finite types), it becomes provable.

\begin{coqfact}[ ][GMP_unit]
	$\GMP_\Unit$ holds.
\end{coqfact}

\begin{proof}
	Given a function $f:\Unit \to \Bool$ with $\neg\neg(\exists x.\, f\,x=\btrue)$ by case analysis on $f\,\star$ we either obtain $f\,\star=\btrue$ and hence $\exists x.\, f\,x=\btrue$ or it is $f\,\star=\bfalse$ and hence $\neg(\exists x.\, f\,x=\btrue)$, contradicting $\neg\neg(\exists x.\, f\,x=\btrue)$.
\end{proof}

As was the case with $\KS$ and $\MP$, it is straightforward to see that a matching pair $\GKS_B$ and $\MP_B$ implies $\LEM$.
Therefore we conclude decompositions generalising \Cref{fact_BDP_LEM}.

\begin{coqfact}[ ][BDP_GMP]
	$\LEM$ is equivalent to both $\BDP^B+\GMP_B$ and $\BDDP^B+\GMP_B$.
\end{coqfact}

\begin{proof}
	Clearly all involved principles are implied by $\LEM$.
	To argue one of the converse directions, by \Cref{fact_BDP_GKS} we know that $\BDDP^B$ implies $\GKS_B$ and obviously the latter together with $\GMP_B$ implies double negation elimination of arbitrary propositions, hence $\LEM$.
\end{proof}

Note that for the corner case $B:=\Unit$ this also presents a refactorisation of \Cref{fact_dp_lem}, explaining in a more fine-grained way how $\LEM$ can be recovered from $\DP$.

Next moving to generalised blurred choice principles, we first record the generalisations of $\BDC$ and $\DDC$ to arbitrary blurring types instead of just $\Nat$:

\begin{align*}
	\BDC_A^B~:=~&\forall R:A\sto A\sto\Prop.\, \total R \to \exists f:B\sto A.\,\total {R\circ f}\\
	\DDC_A^B~:=~&\forall R: A\sto A \sto \Prop.\, \directed R \to \exists f:B\sto A.\, \directed{R\circ f}
\end{align*}

As for other blurred principles, we can immediately observe that $\BDC_A^A$ and $\DDC_A^A$ are always provable.
Regarding decompositions of $\DC$, however, we see no simple generalisations of results like \Cref{larry_bdc2} as the formulation of $\DC$ hinges on the concrete successor function of $\Nat$ and thus cannot directly be generalised to other types.
Moreover, $\BDC_A^B$ and $\DDC_A^B$ seem a bit unnatural, for instance being refutable for finite blurring types $B$.
In contrast, a generalisation of $\BCC$ to arbitrary domains yields a more meaningful blurring of the axiom of choice as follows:
$$\BAC_{A,B}~:=~\forall R:A\sto B\sto \Prop.\,\total R \to \exists f:A\sto B.\forall x.\exists y.\,R\,x\,(f\,y)$$

For instance, indeed $\AC$ splits into a blurring via $\BAC$ and a concretisation of $\AC$:

\begin{coqfact}[ ][AC_BAC]
	$\AC_A$ is equivalent to $\BAC_A+\AC_{A,A}$.
\end{coqfact}

\begin{proof}
	This is an exact generalisation of \Cref{fact_bcc}, with $A$ playing the role of $\Nat$.
\end{proof}

We end this section by anticipating the expected role of $\BAC$ in generalising our results concerning $\DLS$ to cover some uncountable signatures as done by Esp\'indola~\cite{espindolalowenheim} and Karagila~\cite{karagila2014downward}.
First, for a suitable statement of $\DLS$, we restrict to signatures only containing binary predicate symbols, therefore allowing them to be identified with their underlying type $B$ of symbols without loss of information, and that are infinite in the strong sense that $B$ and the type $\Form$ of formulas over $B$ are isomorphic.
For such $B$, we let $\DLS_B$ state that for every model $\M$ there is an elementary embedding $\N\preceq \M$ from a model $\N$ with the type $B$ as domain.
We then obtain the following reverse results:

\begin{coqfact}[ ][DLS_rev]
	\label{fact_BAC}
	Let $B$ be as above. Then $\DLS_B$ implies $\BDP^B$, $\BDDP^B$, $\BDC^B$, $\DDC^B$, and $\BAC_B$.
\end{coqfact}

\begin{proof}
	All follow the by now usual pattern.
	For instance for $\BAC_{B,C}$, let $R:B\to C\to \Prop$ be a total relation.
	We equip $C$ with the structure of a model $\M$ for $B$ seen as a unary signature by setting $P^\M:=R\,P$ for every $P:B$ seen as predicate symbol.
	Then by $\DLS_B$ we obtain a model $\N$ over domain $B$ such that $\N\preceq \M$, witnessed by an elementary embedding $h:B\to C$ that we show to be a blurred choice function in the sense of $\BAC_B,C$ as follows.
	
	First note that by totality of $R$, the model $\M$ satisfies the formulas $\exists \x. P(\x)$ for all $P:B$, which by elementarity is reflected down to $\N$ thus providing some $P':B$ such that $P^\N\,P'$ holds.
	Now by elementarity of $h$ it follows that $P^\M\,(h\,P')$ must hold, therefore $R\,P\,(h\,P')$ as needed to establish $h$ as blurred choice function.
\end{proof}

We leave it for future work to state and prove a forward direction for an exact classification of the strength of $\DLS$ at arbitrary cardinalities.

\section{Separating Blurred Axioms}
\label{sec_separate}

\setCoqFilename{SeparatingAxioms}

In this semi-technical section, we collect evidence for our implicit understanding of the newly introduced axioms: none of the blurred axioms globally imply their original counterpart, none of them overlap in a detectable way, and especially the separation between the fragments of $\AC$ and those of $\LEM$ remains intact.
We only state results as facts, that allow for a mostly internal proof and end in a well-understood underivability of another principle in a similar enough setting -- all other results are stated as conjectures with or without proof sketches.

We begin with the observation that instances of $\BDP$ may be strict weakenings of $\DP$:

\begin{coqfact}[ ][BDP_DP_local]
	$\BDP_A$ does not in general imply $\DP_A$.
\end{coqfact}

\begin{proof}
	$\BDP_\Nat$ is trivially provable while $\DP_\Nat$ implies $\LPO$ (\Cref{fact_lpo}), which is well-known to be independent, for instance invalid in Hyland's effective topos~\cite{hyland1982effective}.
\end{proof}

The next conjecture states an expected strengthening for Boolean blurring.

\begin{coqfact}[ ][BDP_DP_bool]
	$\BDP^\Bool_A$ does not in general imply $\DP_A$.
\end{coqfact}

\begin{proof}
	Note that $\BDP^\Bool_\Bool$ is trivially provable, however $\DP_\Bool$ is not.
	Arguing indirectly, note that $\DP_\Bool$ implies the (unprovable) weak law of excluded middle ($\forall p.\,\neg p \lor \neg \neg p$ ), see~\cite{warren2018drinker}.
	For a direct refutation, consider the following complete Heyting algebra:
	$$\begin{tikzcd}
		& \top \ar[from=dr] &   \\
		a \ar[ur]  
		\ar[from=dr]   &           & b \\
		& \bot \ar[ur] &   \\
	\end{tikzcd}$$
	\vspace{-0.9cm}
	
	We use it to define a predicate $P:\Bool\to \Prop$ with truth values $P\,\btrue:=a$ and $P\,\bfalse:=b$.
	If $\DP_\Bool$ were provable, then there would be a Henkin witness for $P$, say $\btrue$ wlog.
	But then by the Henkin property, $P\,\btrue$, of truth value $a$, would imply $\forall x.\,P\,x$, of truth value $\bot$.
	However, this implication overall evaluates to $b$ in the given algebra and not to $\top$.
\end{proof}

While the previous two statements concern local instances of $\BDP$, the next one expresses the stronger result that also all instances collectively still remain below full classical logic.

\begin{conjecture}
	$\BDP$ does not imply $\MP$ so in particular not $\LEM$.
\end{conjecture}

\begin{proof}[Sketch]
	We expect that $\BDP$ can be validated in realisability models providing a quote operator that allows for reflecting the syntax of formulas and terms.
	In such a setting, all types effectively behave as if they were countable, thus reducing $\BDP$ to the provable case of $\BDP_\Nat$.
	Such quote operators are studied in the context of classical realisability~\cite{krivine2009realizability} but naturally in this setting full classical logic is validated.
	To simultaneously refute $\MP$ and therefore $\LEM$ one would have to either restrict classical realisability to its quoting mechanism, or extend more constructive treatments of quoting like~\cite{pedrot2024upon} to apply to arbitrary types.
\end{proof}

Recall that the analogous statement for $\BDP^\Bool$ was refuted in \Cref{fact_bdp_bool}.
Next, turning to choice principles, we first record a similar result concerning instances of $\BCC$ and $\CC$.

\begin{coqfact}[ ][BCC_CC]
	$\BCC_A$ does not in general imply $\CC_A$.
\end{coqfact}

\begin{proof}
	$\BCC_\Nat$ is trivially provable while $\CC_\Nat$ is well-known to be independent, for instance invalidated in non-deterministic realisability models~\cite{cohen2019effects}.
\end{proof}

Again, we believe that this scales to a global separation of $\BCC$ from $\CC$.

\begin{conjecture}
	\label{conj_BCC}
	$\BCC$ does not imply $\CC_\Nat$, so in particular not $\CC$.
\end{conjecture}

\begin{proof}[Sketch]
	We expect that $\BCC$ holds in the realisability model of~\cite{cohen2019effects}, therefore providing the desired separation.
	The idea is that, while simple non-determinism introduced by a coin flip blocks a realiser for a totality proof for a relation $R:\Nat\to A\to\Prop$ from being turned into a choice function $f:\Nat\to A$, therefore refuting already $\CC_\Nat$, it should be possible to extract a blurred choice function $f:\Nat\to A$ that accumulates all potential computation outcomes, therefore witnessing $\BCC$.
	Additional evidence is provided by a separation result by Castro~\cite{CastroThesis}, who independently identifies and separates an instance of $\BCC$ he calls collection axiom, using classical realisability~\cite{krivine2009realizability}.
\end{proof}

Additionally, our decomposition of $\DC$ relies on the assumption that $\BCC$ and $\DDC$ are not inter-derivable.
The simple direction can be argued as follows:

\begin{coqfact}[ ][BCC_DDC]
	$\BCC$ does not imply $\DDC$.
\end{coqfact}

\begin{proof}
	If this were provable, then via \Cref{larry_bdc2} also $\CC$ would imply $\DC$, which is well-known not to be the case, even for ZF set theory~\cite{jech2008axiom}.
\end{proof}

To date, we are not aware of any strategy for the converse direction, so we just state it here for the sake of completeness.

\begin{conjecture}
	$\DDC$ does not imply $\BCC$.
\end{conjecture}

Note that this would in particular confirm that $\DDC$ does not imply $\DC$.
We expect that this even applies to $\BDC^2$, witnessed by the following strengthening.

\begin{conjecture}
	$\BDC^2$ does not imply $\CC_\Nat$, so in particular not $\DC$.
\end{conjecture}

\begin{proof}[Sketch]
	It is conceivable that $\BDC^2$ holds in the realisability model of~\cite{cohen2019effects} by a similar argument as above in the proof sketch for \Cref{conj_BCC}.
\end{proof}

Finally, concerning the separation of choice principles and fragments of the excluded middle, we observe that axioms up to $\DC$ cannot imply $\BDP$ as an unwanted side effect, even in extensional settings that would trigger the well-known implication from $\AC$ to $\LEM$.

\begin{coqfact}[ ][DC_BDP]
	$\DC$ does not imply $\BDP$, even in an extensional setting.
\end{coqfact}

\begin{proof}
	Hyland's effective topos~\cite{hyland1982effective} is an extensional model satisfying $\DC$ and $\MP$ but invalidating $\LPO$.
	If it were to also satisfy $\BDP$ by the assumed implication, then by \Cref{fact_BDP_LEM} it would also satisfy $\LEM$ and thus in particular $\LPO$.
\end{proof}

Conversely, we can show that $\BCC$ still contains some non-trivial amount of the axiom of choice, therefore not following even from full $\LEM$.

\begin{coqfact}[ ][LEM_BCC]
	$\LEM$ does not imply $\BCC$, even in a setting with unique choice.
\end{coqfact}

\begin{proof}
	As $\LEM$ together with unique choice implies $\CC_\Nat$, the supposed implication together with \Cref{fact_bcc} would mean that $\LEM$ alone proves $\CC$, which is again well-known not to be the case, even for ZF set theory~\cite{jech2008axiom}.
\end{proof}

We end with the only stated conjecture that the merge of all blurred principles into $\OBDC^2$ still does not trivialise either of the orthogonal logical dimensions.

\begin{conjecture}
	$\OBDC^2$ implies neither $\MP$ nor $\CC_\Nat$.
\end{conjecture}

\section{Discussion}
\label{sec_discussion}

\subsection{Main Results}

In this paper, we have studied several logical decompositions of the $\DLS$ theorem over classical and constructive meta-theories.
We briefly summarise the main results as a base for comparison.
First, over a fully classical meta-theory with $\CC_\Nat$ and $\LEM$, we have:
$$\CC_\Nat+\LEM~\vdash~\DLS\leftrightarrow \DC\leftrightarrow \BDC$$

This is the previously known equivalence to $\DC$ (\Cref{larry_LS_classical}), additionally refined by only using $\BDC$ as a blurred weakening of $\DC$ that is equivalent over $\CC_\Nat$ (\Cref{fact_bdc}).

Secondly, assuming just $\CC_\Nat$ in the meta-theory, we obtain:
$$\CC_\Nat~\vdash~\DLS\leftrightarrow \DC + \BDP + \BDDP$$

This explains which fragment of $\LEM$ is needed (\Cref{larry_blurred}), where $\BDP$ and $\BDDP$ independently cover the contribution of syntactic universal and existential quantification.
Again, given $\CC_\Nat$ in the background, $\DC$ could be replaced by any of its blurrings.

Lastly, in a fully constructive meta-theory, so without any base assumptions, we observe:
$$\vdash~\DLS \leftrightarrow \DDC+\BCC+\BDP+\BDDP\leftrightarrow \BDC^2+\BDP+\BDDP$$

This unveils the individual fragments of $\DC$ and $\CC$ needed, namely $\DDC$ and $\BCC$, which together form $\BDC^2$ (\Cref{thm_full}).
Using $\OBDC^2$ that integrates $\BDP$ and $\BDDP$, we finally have:
$$\vdash~\DLS\leftrightarrow \OBDC^2$$

These decompositions provide a clear logical characterisation of the $\DLS$ theorem and the observed principles appear naturally:
same as the $\DLS$ theorem, they all in one way or another collapse arbitrary to countable cardinality.

\subsection{General Remarks}

The central theme governing the results in this paper is the idea of blurring.
Most directly, it appears in the weakening of the drinker paradoxes to hide information of classical existential quantification.
Thereby, the obtained hierarchies $\BDP^B$ and $\BDDP^B$ for different blurrings $B$ are natural generalisations of $\DP$ and $\DDP$ and we expect that already countable blurring $\BDP^\Nat$ and $\BDDP^\Nat$ has a constructively weaker status.
Relatedly, the blurred versions of choice axioms unveil interesting structure, for instance by explaining that $\DC$ without the contributions of $\CC$ states that arbitrary relations with some first-order expressible property must admit countable sub-relations of the same property, as exemplified by $\BDC$ and $\DDC$.

Our proof strategy to use variable environments to represent syntactic submodels seems to be an alternative to the usual strategies to either extend the signature~\cite{marker2006model} or the submodel as a subset of the original model~\cite{bodirsky_notes}.
After the extension process has reached a fixed point, we simply turn the obtained (blurred) Henkin environment into its induced syntactic submodel, where the Henkin property is reminiscent of the Tarski-Vaught test~\cite{simpson_notes} usually applied to check elementarity.
We are not aware of another proof following our strategy and deem it advantageous for our constructive analysis and particularly suitable for mechanisation.

Another point to mention is that we do not incorporate equality as a primitive of the syntax and thereby external cardinality of a model and its internal cardinality based on equivalence classes of first-order indistinguishability need not coincide.
This choice allows us to state the $\DLS$ theorem more generally, applying to all and not just infinite models, and subsumes the traditional presentation with equality, as the internal cardinality of a model is bounded by the external cardinality.
Connectedly, we use the wording for ``countability'' to include finite cardinality, such that we do not have to talk of ``at most countable'' models.

\subsection{Rocq Mechanisation}

The Rocq development accompanying this paper is based on and contributed to the Rocq library of first-order logic~\cite{kirst2022library}.
This library provides the core definitions of syntax, deduction systems, and semantics, as well as a constructive completeness proof we build on for our first approximation of the $\DLS$ theorem (\Cref{fact_LS_compl}).
The handling of variables is done in the style of the Autosubst 2 framework~\cite{AutoSubst2}, employing parallel substitutions for de Bruijn indexed syntax and providing a normalisation tactic for substitutive expressions.
On top of that library, our development spans roughly 3,500 lines of code, with only around 300 needed for a self-contained proof of the $\DLS$ theorem.
The latter illustrates that our proof strategy based on variable environments instead of signature or model extension is indeed well-suited for computer mechanisation.

We are aware of a few other mechanisations of the $\DLS$ theorem.
In Isabelle/HOL, Blanchette and Popescu~\cite{blanchette2013mechanizing} give a classical and mostly automated proof of the limited strength of our \Cref{fact_LS_compl}, as by-product of a Henkin-style completeness proof.
Using Mizar, Caminati~\cite{caminati2010basic} also proves the weak form of the $\DLS$ theorem corresponding to our \Cref{fact_LS_compl}, again following the strategy factoring through a classical completeness proof.
Contained in the Lean mathematical library~\cite{The_mathlib_Community_2020} and contributed by Anderson is a classical proof of the $\DLS$ theorem in strong form, i.e.~providing an elementary submodel.
This proof relies on the full axiom of choice to obtain Skolem functions for arbitrary formulas.

\subsection{Future Work}
\label{sec_future}

For the purpose of this paper, we have focused on the case of countable signatures only.
As discussed by Esp\'indola~\cite{espindolalowenheim} and Karagila~\cite{karagila2014downward}, the classical equivalence of the $\DLS$ theorem to $\DC$ generalises to signatures of higher cardinality: for signatures of size $A$, one needs $\AC_A$ on top of $\DC$, which was not visible in the case $A:=\Nat$ since $\AC_\Nat$, that is $\CC$, happens to follow from $\DC$.
As anticipated in \Cref{sec_general}, we conjecture that, in our constructive setting, something similar can be observed, namely that we need the following assumptions: $\DDC$ as before, $\BDP^A$ and $\BDDP^A$ now blurred with respect to $A$, and, in replacement of $\BCC$, a blurred form of the general axiom of choice akin to $\BAC_A$.
In \Cref{fact_BAC}, we have already verified that a suitable variant of the $\DLS$ theorem at signature size $A$ implies $\BDP^A$, $\BDDP^A$, and $\BAC_{A,B}$ for all $B$, but whether they together in turn imply $\DLS$ in this or an alternative formulation is left for future work.
Especially, this proof would require a more conventional proof strategy since our trick to use variable environments, with $\Nat$ as domain, to represent submodels, now with $A$ as domain, is certainly not applicable.

Another interesting direction is to consider the upwards case of the Löwenheim-Skolem theorem, stating that every infinite model has an elementary extension at arbitrarily larger cardinality.
For this statement, contrarily to the downwards case, syntactic equality is crucial to classify the actual internal cardinality of the extended model.
While without this restriction, as we already mechanised the proof is rather trivial and fully constructive by just adding enough new elements, with the restriction,
the proof usually employs the compactness theorem to ensure the distinctness of newly added elements to increase the cardinality.
The compactness theorem, however, is known to not be constructive, leaving the constructive status of the upwards Löwenheim-Skolem theorem to be investigated.

Finally, our working hypothesis regarding the status of the blurred logical principles is that neither of them collapses, i.e.~most crucially that $\BDP+\BDDP$ does not imply $\LEM$, that $\BCC$ does not imply $\CC$, that $\DDC$ does not imply $\BCC$, and that $\BDC$ does not imply $\BDC^2$.
Recall from \Cref{sec_separate} that most implications do not hold locally, i.e.~it is not the case that for instance $\BCC_A$ always implies $\CC_A$, as $\BCC_\Nat$ is provable while $\CC_\Nat$ is not.
Therefore, only the question of global implications remains of interest and to obtain full certainty, one has to construct separating models.
A promising approach as sketched in \Cref{sec_separate} is the use of realisability models, where the logical components are interpreted by computational means.
For instance, non-deterministic realisability allows to invalidate $\CC$ (and thus also $\DC$)~\cite{cohen2019effects} and there is hope that in this setting still $\BCC$ (and maybe even $\BDC$) are validated.
Moreover, based on classical realisability, Castro~\cite{CastroThesis} independently identifies and separates an instance of $\BCC$ he calls collection axiom.
Finally, realisability models incorporating quoting operations~\cite{pedrot2024upon} might be useful to separate $\BDP$ and $\BDDP$ from $\LEM$.
However, all sketches along these lines reported in \Cref{sec_separate} are far from formal proofs that would require careful translation of existing results to the concrete case of Rocq's type theory CIC, ideally supported by computer mechanisation.

\section*{Acknowledgments}

We thank Valentin Blot, Yannick Forster, Hugo Herbelin, Felix Jahn, Julian Rosemann, and the anonymous reviewers for their helpful comments on drafts of this paper.
In addition, we thank Martin Baillon, F\'elix Castro, Thierry Coquand, Jonas Frey, Guillaume Geoffroy, Sam van Gool, Timothée Huneau, Jean-Louis Krivine, Jérémie Marquès, Alexandre Miquel, \'Etienne Miquey, Pierre-Marie P\'edrot, Quentin Schroeder, and Joshua Wrigley for multiple discussions about the blurred axioms in relation to their original forms.

\bibliographystyle{alphaurl}
\bibliography{biblio_LMCS}

\appendix

\newpage

\section{Overview of Logical Principles}
\label{app_defs}

Standard principles below the excluded middle:\label{defs_lem}
\begin{align*}
\LEM ~:=~& \forall p:\Prop.\,p\lor \neg p\\
\LPO ~:=~& \forall f:\Nat\sto\Bool.\,(\exists n.\, f\,n=\btrue)\lor (\forall x.\,f\,n=\bfalse)\\
\DP_A ~:=~& \forall P:A\sto\Prop.\,\exists x.\, P\,x \to \forall y.\, P\,y\\
\DDP_A ~:=~& \forall P:A\sto\Prop.\,\exists x.\, (\exists y.\,P\,y) \to P\,x\\
\IP_A ~:=~& \forall P:A\sto\Prop.\,\forall p:\Prop.\,(p\to \exists x.\,P\,x)\to \exists x.\,p\to P\,x\\
\KS~:=~&\forall p:\Prop.\exists f:\Nat\sto\Bool.\,p\leftrightarrow \exists n.\,f\,n=\btrue\\
\MP ~:=~&\forall f:\Nat\sto\Bool.\,\neg\neg(\exists n.\, f\,n=\btrue)\to \exists n.\, f\,n=\btrue
\end{align*}

\vspace{0.3cm}
\noindent
Generalised principles below the excluded middle:
\begin{align*}
	\GKS_A~:=~&\forall p:\Prop.\exists f:A\sto\Bool.\,p\leftrightarrow \exists x.\,f\,x=\btrue\\
	\GMP_A ~:=~& \forall f:A\sto\Bool.\,\neg\neg(\exists x.\, f\,x=\btrue)\to \exists x.\, f\,x=\btrue
\end{align*}

\vspace{0.3cm}
\noindent
Standard principles below the axiom of choice:\label{defs_ac}
\begin{align*}
\AC_{A,B} ~:=~& \forall R: A\sto B\sto\Prop.\, \total R \to \exists f:A\sto B. \forall x.\, R\,x\,(f\,x)\\
\DC_{A} ~:=~&\forall R:A\sto A \sto\Prop.\,\total R\to \exists f:\Nat\sto A.\forall n.\,R\,(f\,n)\,(f\,(n+1))\\
\CC_{A} ~:=~& \forall R: \Nat\sto A\sto\Prop.\, \total R \to \exists f:\Nat\sto A. \forall n.\, R\,n\,(f\,n)\\
\OAC_{A,B} ~:=~& \forall R: A\sto B\sto\Prop.  \exists f:A\sto B. \,\total R \to \forall x.\, R\,x\,(f\,x)
\end{align*}

\vspace{0.3cm}
\noindent
Blurred principles below the excluded middle:\label{defs_blem}
\begin{align*}
\BDP^B_A~:=~&\forall P:A\sto\Prop.\exists f:B\sto A.\,(\forall y.\,P\,(f\,y))\to \forall x.\, P\,x\\
\BDDP^B_A~:=~&\forall P:A\sto\Prop.\exists f:B\sto A.\,(\exists x.\, P\,x)\to\exists y.\,P\,(f\,y)\\
\BIP_A^B ~:=~ &\forall P:A\sto\Prop.\,\forall p:\Prop.\,(p\to \exists x.\,P\,x)\to \exists f:B\sto A.\,p\to \exists y.\,P\,(f\,y)
\end{align*}

\vspace{0.3cm}
\noindent
Blurred principles below the axiom of choice:\label{defs_bac}
\begin{align*}
	\BCC_A~:=~&\forall R:\Nat\sto A\sto \Prop.\,\total R \to \exists f:\Nat\sto A.\forall n.\exists m.\,R\,n\,(f\,m)\\
	\BDC_A~:=~&\forall R:A\sto A\sto\Prop.\, \total R \to \exists f:\Nat\sto A.\,\total {R\circ f}\\
	\BDC^2_A~:=~&\forall R:A^2\sto A\sto\Prop.\, \total R \to \exists f:\Nat\sto A.\,\total {R\circ f}\\
	\DDC_A~:=~&\forall R: A\sto A \sto \Prop.\, \directed R \to \exists f:\Nat\sto A.\, \directed{R\circ f}\\
	\OBDC^2_A~:=~&\forall R:A^2\sto A\sto\Prop. \exists f:\Nat\sto A.\,\total R \leftrightarrow\total {R\circ f}
\end{align*}

\vspace{0.3cm}
\noindent
Generalised blurred principles below the axiom of choice:
\begin{align*}
	\BDC_A^B~:=~&\forall R:A\sto A\sto\Prop.\, \total R \to \exists f:B\sto A.\,\total {R\circ f}\\
	\DDC_A^B~:=~&\forall R: A\sto A \sto \Prop.\, \directed R \to \exists f:B\sto A.\, \directed{R\circ f}\\
	\BAC_{A,B}~:=~&\forall R:A\sto B\sto \Prop.\,\total R \to \exists f:A\sto B.\forall x.\exists y.\,R\,x\,(f\,y)
\end{align*}

\newpage
\section{Connections of Logical Principles}
\label{app_imps}
See below for an overview of our main results regarding $\DLS$.
As before, solid arrows depict (presumably strict) implications while the dashed arrows depict combined equivalences.
Moreover, double arrows depict direct equivalences with potential side conditions placed next to the arrows.


\vspace{1cm}
\hspace{-1cm}
\adjustbox{scale=1.169}{
\begin{tikzcd}[sep=huge]
	&&& {\DC} \\
	& \BDP \\
	{\OBDC^2} & {\BDC^2} && {\DDC+\BCC} & {\DLS} \\
	& \BDDP
	\arrow[from=3-1, to=2-2]
	\arrow[from=3-1, to=4-2]
	\arrow[from=3-1, to=3-2]
	\arrow[curve={height=42pt}, Rightarrow, 2tail reversed, from=3-1, to=3-5]
	\arrow[dashed, from=4-2, to=3-5]
	\arrow[dashed, from=2-2, to=3-5]
	\arrow[dashed, tail reversed, no head, from=3-5, to=3-4]
	\arrow["{\CC_\Nat + \LEM}", Rightarrow, 2tail reversed, from=1-4, to=3-5]
	\arrow[from=1-4, to=3-4]
	\arrow["{\CC_\Nat}", Rightarrow, 2tail reversed, from=3-2, to=1-4]
	\arrow[Rightarrow, 2tail reversed, from=3-2, to=3-4]
\end{tikzcd}
}

\vspace{2cm}
For convenience, we also reproduce the diagrams from \Cref{sec_BDP,sec_BAC}.

\vspace{1cm}
\hspace{-1cm}
\adjustbox{scale=1.169}{
\begin{tikzcd}[column sep=3.15em,row sep=large]
	& {\LEM/ \DP/ \DDP/ \IP} \\
	\BDP & \LPO & \BDDP \\
	{\KS'} & \MP & \KS
	\arrow[from=1-2, to=2-2]
	\arrow[from=1-2, to=2-1]
	\arrow[from=1-2, to=2-3]
	\arrow[from=2-1, to=3-1]
	\arrow[from=2-2, to=3-2]
	\arrow[from=2-3, to=3-3]
	\arrow[""{name=0, anchor=center, inner sep=0}, dashed, no head, from=3-2, to=3-3]
	\arrow[""{name=1, anchor=center, inner sep=0}, dashed, no head, from=3-1, to=3-2]
	\arrow[curve={height=12pt}, shorten <=9pt, dashed, from=0, to=1-2]
	\arrow[curve={height=-12pt}, shorten <=9pt, dashed, from=1, to=1-2]
\end{tikzcd}

\begin{tikzcd}[column sep=3.15em,row sep=large]
	\DC && \CC \\
	{\BDC^2} & \BDC & \BCC \\
	\DDC
	\arrow[from=1-1, to=2-1]
	\arrow[from=2-1, to=3-1]
	\arrow[from=1-1, to=1-3]
	\arrow[from=1-3, to=2-3]
	\arrow[from=2-1, to=2-2]
	\arrow[from=2-2, to=2-3]
	\arrow[""{name=0, anchor=center, inner sep=0}, dashed, no head, from=3-1, to=2-3]
	\arrow[""{name=1, anchor=center, inner sep=0}, dashed, no head, from=2-1, to=1-3]
	\arrow[shorten <=5pt, dashed, from=0, to=2-1]
	\arrow[shorten <=6pt, dashed, from=1, to=1-1]
\end{tikzcd}

}
\end{document}